\documentclass[reqno]{amsart}

\usepackage{amsmath,amssymb,amsthm,amsfonts}
\usepackage{ulem}
\usepackage{color}
\usepackage{hyperref}
\usepackage{color}

\usepackage{float}
\usepackage{graphicx}

\hypersetup{colorlinks=true,linkcolor=blue,citecolor=blue}
\numberwithin{equation}{section}

\newtheorem{theorem}{Theorem}[section]
\newtheorem{lemma}[theorem]{Lemma}

\newtheorem{cor}[theorem]  {Corollary}

\theoremstyle{definition}

\theoremstyle{remark}
\newtheorem{remark}{Remark}

\numberwithin{remark}{section}

\newcommand{\z}{\mathbf{0}}
\newcommand{\g}{\boldsymbol{\gamma}}
\newcommand{\e}{\boldsymbol{\eta}}
\newcommand{\s}{\lesssim}
\newcommand{\probBC}{\mathbb{P}^{\z}_{\Lambda,\mu_0}} %probability with eta boundary condition and mu 0 
\newcommand{\canBC}{Z^{\z}_{\Lambda,\beta}} % Canonical P.F. with eta boundary condition
\newcommand{\GcanBC}{\Xi^{\z}_{\Lambda,\beta}} % Grand-Canonical P.F. with eta boundary condition
\newcommand{\tN}{\tilde{N}} % N tilde
\newcommand{\bN}{\bar{N}_{\Lambda}} % N bar
\newcommand{\bRL}{\bar{\rho}_{\Lambda}} %bar density finite volume
\newcommand{\tRL}{\tilde{\rho}_{\Lambda}} %tilde d.f.v.
\newcommand{\F}{\mathcal{F}_{\Lambda,\beta,\z}}
\newcommand{\fgc}{f^{GC}_{\Lambda,\beta,\z}}
\newcommand{\sig}{\boldsymbol{\sigma}}

\newcommand{\be}{\begin{equation}}
\newcommand{\ee}{\end{equation}}
\newcommand{\ba}{\begin{equation} \begin{aligned}}
\newcommand{\ea}{\end{aligned}\end{equation}}
\newcommand{\bes}{\begin{equation*}}
\newcommand{\ees}{\end{equation*}}

\def\1{{\mathchoice {1\mskip-4mu\mathrm l}      % Blackboard bold 1
		{1\mskip-4mu\mathrm l}
		{1\mskip-4.5mu\mathrm l} {1\mskip-5mu\mathrm l}}}
\begin{document}
	
	%%%%%%%%%%%%%%
	% Front matter %% 
	\title{Cluster expansion for the Ising model in the canonical ensemble}
    \author{Giuseppe Scola}
    	%{Gran Sasso Science Institute (GSSI), \\ Viale Francesco Crispi, 7, 67100, L'Aquila (IT).\\ 
    	%\email{giuseppe.scola@gssi.it}
	\date{}
	\maketitle
	%%%%%%%%
	\begin{abstract} 	We show the validity of the cluster expansion in the canonical ensemble for the Ising model.
		We compare the lower bound of its radius of convergence with the one computed by the virial expansion working
		in the grand-canonical ensemble.
		Using the cluster expansion we give direct proofs with quantification of the higher order error terms for the decay of correlations, central limit theorem and large deviations.\\
		
		\noindent\emph{2020 Mathematics Subject Classification}: {60F05, 60F10, 82B05, 	82B20}
		
		 \noindent\emph{Keywords}: {Ising model, lattice system, cluster expansion, decay of correlations, virial expansion, precise large deviations, local moderate deviations}

	\end{abstract}
	
	\tableofcontents
	%%%%%%%%%%%%% 

\section{Introduction}

The analysis of the relation between thermodynamic quantities and their use for quantitative prediction of macroscopic properties of matter through its microscopic structure, is the fundamental goal of statistical mechanics. A key tool in this direction is the {\it{cluster expansion}} initially developed by Mayer \cite{mayerstatistical} for non-ideal gases viewed as a perturbation around the ideal gas. 
This method allows to express the thermodynamic quantities as absolutely convergent power series.
Over the last years many methods and generalizations have been developed mostly adapted
to the grand-canonical ensemble
%The conditions of convergence of this series with respect to  temperature, activity, density and interactions between the particles, are very well studied both for the {\it{canonical}} and {\it{grand-canonical ensemble}}, \cite{brydges1984short}, \cite{jansen2019virial}, \cite{morais2013continuous}, \cite{penroseconvergence}, \cite{procacci2017convergence},  \cite{pulvirenti2012cluster}. 
as the canonical constraint seemed restrictive to the product structure of the underlying system.
As it was proved in \cite{pulvirenti2012cluster} one can remove this constraint by viewing it as a hard-core system
of clusters
and then it fits beautifully in the existing theory of the abstract polymer model \cite{fernandez2007cluster}, \cite{gruber1971general}, \cite{kotecky1986cluster}.
For the case of the Ising model, the constraint seems even more restrictive as
particles are indexed by their lattice position, with the constraint of fixed magnetization
destroying the product structure.
Hence, the key points of the present paper are the followings:

\begin{enumerate}
	\item We view the Ising model as a lattice gas \cite{farrell1966cluster} 
	and by indexing the spins (as in the continuous case) rather than their position, we can treat the canonical constraint in a similar manner as in \cite{pulvirenti2012cluster}. 
	\item We use the representation in point 1. to establish the condition of convergence working in the grand-canonical ensemble. Moreover, we compare it with the one obtained using the contour representation for the Ising model (see for example \cite{friedli2017statistical}). Furthermore, we compare the  lower bound of the radius of convergence  for the virial expansion - in density - working in the grand-canonical ensemble with the  convergence condition obtained in the canonical ensemble, as in point 1.
	\item As a by-product of point 1., we estimate the decay of correlations working directly in the canonical ensemble.
	\item We prove moderate and large deviations with quantification of the higher order error terms
	and compare with the results developed previously in \cite{del1974local} and \cite{dobrushin1994large}
	for the case of the Ising model in the grand-canonical ensemble.
	%\item 
	\item It is worth noticing that despite the fact that we present the method for the case of the Ising model, we expect that our approach is applicable to more general lattice systems with more complicated interactions.
\end{enumerate}

%In the present paper, taking up the idea of \cite{pulvirenti2012cluster} 
%we investigate these conditions for the ferromagnetic Ising model in the canonical ensemble by viewing it as a lattice gas and using tools from 
%\cite{farrell1966cluster} and \cite{morais2013continuous}. What appears clear from our approach is that the most suitable way to apply the cluster expansion technique here is to pass from the spins representation to an interacting particles system on the lattice. This can be done deriving from the Ising model an appropriate \textquotedblleft hard core type potential\textquotedblright (i.e. allowing the possibility of overlapping between the particles) and founding also a proper relation between the density of the second  representation and the {\it{magnetization}} of the starting model. Furthermore one can use this rewriting also in the grand-canonical ensemble finding a different representation from the classical one as it  is presented for example in  \cite{dobrushin1994large} and \cite{friedli2017statistical}. 

%Moreover the cluster expansion of the canonical partition function lends itself to some interesting applications as the analysis of the behavior of the {\it{2-point correlation function}} (as done in \cite{kuna2018convergence}), or the study of Precise Large and Local Moderate Deviations with the grand-canonical probability measure \cite{scola2020local}. In particular this second application gives a new practical and direct way to address a well known classical probability problem.   

The structure of the paper is the following: in Section~\ref{Sec1}, we present the model and results. The main theorem about the validity of the cluster expansion for the canonical partition function under an appropriate condition on the density (Theorem \ref{TH1}) is given in Subsection \ref{SubSec1}. Its proof is in Section \ref{SectionCE}. The decay of the 2-point correlation for the canonical ensemble and the application of Theorem \ref{TH1} to Precise Large Deviations and Local Moderate Deviations are presented in Subsections \ref{SubSec1} and \ref{SubSec3} (respectively Theorems \ref{TH2}, \ref{TH-LD}, \ref{TH-MD} and Corollary \ref{CLT}). The proof of Theorem \ref{TH2} (2-point correlation) is in Section \ref{SectionD}. In Section \ref{SectionRD}, we show an analysis in the grand-canonical ensemble in which we compare the approach presented in \cite{friedli2017statistical} with the one presented in \cite{farrell1966cluster} and analyzed in detail in the present paper. Moreover, applying \cite{jansen2019virial}, we give the virial inversion in the grand-canonical ensemble, and we compare the resulting  - density - convergence condition with the one obtained in the canonical ensemble (Lemma \ref{LemmaCluster}). We conclude with Section \ref{SectionLMD} where we compare our approach for the study of Precise Large Deviations and Local Moderate Deviations, as presented in Theorems \ref{TH-LD}, \ref{TH-MD} and Corollary \ref{CLT}, with the one in the grand-canonical ensemble from \cite{del1974local} and \cite{dobrushin1994large}. For completeness of the exposition, the proofs of the Theorems \ref{TH-LD}, \ref{TH-MD} and Corollary \ref{CLT} (which follows from \cite{scola2020local}), are given in Appendix \ref{S2}.

\section{Notation, model and results}\label{Sec1}

\subsection{Cluster expansion in canonical ensemble }\label{SubSec1}
We consider the ferromagnetic Ising Model on a finite volume $\Lambda\subset\mathbb{Z}^d$ at small inverse temperature $\beta$.   
%The spins are defined as
%\begin{equation}\begin{split}
%\sigma:&\;\mathbb{Z}^d\to\{-1,1\}\\
%&\;\;\; x\mapsto\sigma(x)\nonumber
%\end{split}
%\end{equation}
We denote with $\sig=(\sigma(x_1),...,\sigma(x_{|\Lambda|}))\in\{-1,1\}^{\Lambda}$ a spins vector on $\Lambda$ and with $\sig^c\in\{-1,1\}^{\Lambda^c}$ a spins vector on $\Lambda^c:=\mathbb{Z}^d\setminus\Lambda$ with fixed value $\sigma^c$, i.e.,  such that $\sigma(x)=\sigma^c$ for all $x\in\Lambda^c$. Hence, defining $\mathcal{E}_{\Lambda}:=\{\{x,x'\}\subset\mathbb{Z}^d\;|\;\{x,x'\}\cap\Lambda\ne\emptyset,\;|x-x'|=1\}$, where $|\cdot|$ is the Euclidean distance, the Hamiltonian is given by 
\begin{equation}
	\mathcal{H}_{\Lambda}^{\sig^c}(\sig):=-J \sum_{\{x,x'\}\in\mathcal{E}_{\Lambda}}\sigma(x)\sigma(x'),
	\label{HamIs}
\end{equation}
with $J\in\mathbb{R}^+$. 
The \textit{canonical partition function} at fixed magnetization $m\in( -1,-1+\epsilon),\;0<\epsilon<1$, is defined as  follows:
\begin{equation}
	\tilde{Z}^{\sig^c}_{\Lambda,\beta}(m):=\sum_{\substack{\sig\in\{-1,1\}^{\Lambda}\;:\\\frac{1}{|\Lambda|}\sum_{x\in\Lambda}\sigma(x)=m}} e^{-\beta\mathcal{H}^{\sig^c}_{\Lambda}(\sig)}.
	\label{CpfIs}
\end{equation}

We reduce the canonical partition function for the Ising model given by \eqref{CpfIs} to the one for a classic lattice gas system using the transformation:
\begin{equation}
	\sigma(x)=2\eta(x)-1,
	\label{transf}
\end{equation}
with $\eta:\mathbb{Z}^d\mapsto\{0,1\}$.  Then \eqref{HamIs} and \eqref{CpfIs} become
\begin{eqnarray}
	\mathcal{H}_{\Lambda}^{\sig^c}(\sig)\equiv\mathcal{H}_{\Lambda}^{\e^c}(\e):=4Jm'|\mathcal{E}_{\Lambda}|-J|\mathcal{E}_{\Lambda}|-4J\sum_{\{x,x'\}\in\mathcal{E}_{\Lambda}}\eta(x)\eta(x')
	\label{Ham2}
\end{eqnarray}
and
\begin{equation}
	\tilde{Z}^{\sig^c}_{\Lambda,\beta}(m)\equiv\tilde{Z}^{\e^c}_{\Lambda,\beta}(m'):=\sum_{\substack{\boldsymbol{\eta}\in\{0,1\}^{\Lambda}\;:\\\sum_{x\in\Lambda}\eta(x)=m'|\Lambda|}}e^{-\beta \mathcal{H}^{\e^c}_{\Lambda}(\e)},
	\label{CpfLg}
\end{equation}
where $m':=(m+1)/2$.

We denote with $N\equiv N(m'):=m' |\Lambda|$ the number of (indistinguishable) particles of the system and with $\mathbf{x}=(x_1,...,x_N)\in\Lambda^N$ a configuration vector. We also introduce the following ``hard-core"  potential \cite{farrell1966cluster}:

\begin{equation}
	V(x-x'):=\begin{cases}
		\;\infty\;\;\;\;\;\,\;\mathrm{if}\;x=x',\\
		\;-4J\;\;\;\mathrm{if}\;|x-x'|=1,\\
		\;0\;\;\;\;\;\;\;\,\mathrm{otherwise},
	\end{cases}
	\label{Potential}
\end{equation}
for all $x,x'\in\mathbb{Z}^d$. In this way, from \eqref{CpfLg} we can write
\begin{eqnarray}
	\tilde{Z}^{\e^c}_{\Lambda,\beta}(m')=\exp\left\{-\beta|\Lambda|\left[4Jm'\frac{|\mathcal{E}_{\Lambda}|}{|\Lambda|}-J\frac{|\mathcal{E}_{\Lambda}|}{|\Lambda|}\right]\right\}Z^{\g}_{\Lambda,\beta}(N),
	\label{Cpf}
\end{eqnarray}
where 
\begin{equation}
	Z^{\g}_{\Lambda,\beta}(N):=\frac{1}{N!}\sum_{\mathbf{x}\in\Lambda^N}e^{-\beta H_{\Lambda}^{\g}(\mathbf{x})}.
	\label{CanPF}
\end{equation}
The Hamiltonian $H_{\Lambda}^{\g}(\mathbf{x})$ is given by
\begin{eqnarray}\label{NewH}
	H_{\Lambda}^{\g}(\mathbf{x})&:&(\mathbb{Z}^d)^N\longrightarrow\mathbb{R}\nonumber\\
	&&\;\;\;\mathbf{x}\;\;\;\;\;\mapsto \sum_{\substack{1\le i <j\le N}}V(x_i-x_j)+\sum_{\substack{1\le i\le N,\;j\ge1}}V(x_i-\gamma_j),
\end{eqnarray}
where $\g=(\gamma_1,..,\gamma_i,...)$ is an appropriate fixed configuration  with $\gamma_i\in\Lambda^c$, for all $i\ge1$. In order to simplify the calculation we consider zero boundary conditions such that our Hamiltonian is given by $H^{\mathbf{0}}_{\Lambda}(\mathbf{x})=\sum_{1\le i<j\le N}V(x_i-x_j)$. 

The potential defined in \eqref{Potential} satisfies the usual regularity and stability conditions needed for the cluster expansion.  Indeed, for all fixed $x^*\in\mathbb{Z}^d$ we get
\begin{equation}
	\sum_{1\le j\le N}V(x^*-x_j)\ge-4J\sum_{1\le j\le N}\mathbf{1}_{\{|x^*-x_j|=1\}}(x_j)\ge-8Jd=:-B,
	\label{Stability}
\end{equation}
and
\begin{equation}
	\sum_{x\in\mathbb{Z}^d} \left| e^{-\beta V(x^*-x)}-1\right|=\sum_{\substack{x\in\mathbb{Z}^d\;:\\|x^*-x|\le 1}}\left| e^{-\beta V(x^*-x)}-1\right|=: C_{J,d}(\beta),
\end{equation}
where
\begin{equation}
	C_{J,d}(\beta)=2d(e^{4\beta J}-1)+1 < \infty
	\label{TildeC}
\end{equation}
for all finite $\beta\ge0$.

Defining the {\it{finite volume free energy}} as
\begin{equation}
	f_{\beta,\Lambda,\z}(N):=-\frac{1}{\beta|\Lambda|}\log Z_{\Lambda,\beta}^{\z}(N),
	\label{FVfreeE}
\end{equation}
the {\it{thermodynamic free energy}} is given by
\begin{equation}\begin{split}
		f_{\beta}(\rho):=\lim_{\substack{\Lambda\to\mathbb{Z}^d\\N/|\Lambda|\to\rho}}f_{\beta,\Lambda,\z}(N),
	\end{split}
	\label{FreeE}
\end{equation}
with $\rho>0$ and where the limit is in the sense of Van Hove.

The main result of this paper is the cluster expansion of \eqref{CanPF} presented in Theorem \ref{TH1} below. Thanks to it we also derive an expression for the thermodynamic free energy as an absolutely convergent power series with respect to the density. The coefficients of this expansion are given by the {\it{irreducible (2-connected) Mayer's coefficients}}.  These are defined as 
\begin{equation}
	\beta_n:=\frac{1}{n!}\sum_{\substack{g\in\mathcal{B}_{n+1}\\V(g)\ni\{1\}}}\sum_{(x_2,\ldots,x_{n+1})\in(\mathbb{Z}^d)^n}\prod_{\{i,j\}\in E(g)}(e^{-\beta V(x_i-x_j)}-1),\;\;\;x_1\equiv 0,
	\label{Mayer}
\end{equation}where $\mathcal{B}_{n+1}$ is the set of the graphs with $n+1$ vertices which remain connected when a vertex is removed (called also 2-connected), and $E(g)$ and $V(g)$ are respectively the set of edges and vertices of a graph $g$. Note that these are the \textquotedblleft discrete version\textquotedblright\space of the ones given in \cite{mayerstatistical} - formula (13.25) - where instead of the sum over  $\mathbf{x}\in(\mathbb{Z}^d)^n$ one has the integral over   $(\mathbb{R}^d)^n$.

\begin{theorem}
	There exists a constant $\mathcal{R}_C\equiv\mathcal{R}_C(d,J,\beta)$ independent on $N$ and $\Lambda$ (see Lemma \ref{LemmaCluster} for the precise value), such that if $N/|\Lambda|<\mathcal{R}_C$ then 
	\begin{equation}
		\frac{1}{|\Lambda|}\log Z^{\z}_{\Lambda,\beta}(N)=\frac{1}{|\Lambda|}\log\frac{|\Lambda|^N}{N!}+\frac{N}{|\Lambda|}\sum_{n\ge1}F_{\beta,N,\Lambda}(n),\;\;\;\;
		\label{CeCan}
	\end{equation}
	where $F_{\beta,N,\Lambda}(n)$ is explicitly given \eqref{F_bNL}.
	For this function there exist  constants $C,c>0$ such that for every $N$ and $\Lambda$ and for all $n\ge1$:
	\begin{equation}
		|F_{\beta,N,\Lambda}(n)|\le Ce^{-cn}.
		\label{exponential_decay}
	\end{equation}
	Furthermore, in the thermodynamic limit
	\begin{equation}
		\lim_{\substack{\Lambda\to\mathbb{Z}^d\\N/|\Lambda|\to\rho}}\frac{N}{|\Lambda|}F_{\beta,N,\Lambda}(n)=\frac{1}{n+1}\rho^{n+1}\beta_n,
	\end{equation} 
	for all $n\ge1,\rho<\mathcal{R}_C$ and where $\beta_n$ is given by \eqref{Mayer}. 
	\label{TH1}
\end{theorem}

\begin{remark} 
	For the original formulation of the Ising model (as it is given by partition function \eqref{CpfIs}), the theorem above corresponds to the expansion around magnetization $m=-1$. This results from \eqref{transf}-\eqref{NewH} and from the fact that $N$ is defined via the magnetization ($N=m'|\Lambda|$). Moreover, by symmetry, we can consider a similar expansion around $m=1$ applying $\sigma(x)=1-2\eta(x).$
\end{remark}
\begin{remark}\label{Remark2}
	We recall that in the literature the cluster expansion for the Ising model - in the gran-canonical ensemble - is done by using a polymer model representation usually called {\it{contour expansion}} \cite{friedli2017statistical}. In Section \ref{SectionRD} we will recall it and we compare this expansion with the grand-canonical version of the one given in Theorem \ref{TH1} (for which a first non-rigorous formulation is given in \cite{farrell1966cluster}). Moreover, again in Section \ref{SectionRD}, we will use the latter to derive the {\it{virial inversion}} and we compare the lower bound of its radius of convergence with the one given in the theorem above.
\end{remark}	
\begin{remark}\label{Remark3}
	As it is explained in Subsection \ref{RemarKacBC}, Theorem \ref{TH1} holds true if we assume that the potential \eqref{Potential} acts when $|x-x'|\le R$ with $R^d<|\Lambda|$ (for example in the case of Kac potential) as well as if we consider 
	boundary conditions $\g\ne\z$. 
\end{remark}

As we will see in  Section \ref{SectionCE} the term $F_{\beta,N,\Lambda}(n)$ is given by
\begin{equation}
	F_{\beta,N,\Lambda}(n)=\frac{1}{n+1}P_{N,|\Lambda|}(n)B_{\Lambda,\beta}(n),
	\label{F_bNL}
\end{equation}
where
\begin{equation}
	P_{N,|\Lambda|}(n):=\begin{cases}
		\frac{(N-1)\cdot\cdot\cdot(N-n)}{|\Lambda|^n}\;\;\;{\rm{if}}\;n<N,\\
		0\;\;\;\;\;\;\;\;\;\;\;\;\;\;\;\;\;\;\;\;\;\;{\rm{otherwise}},
	\end{cases}
	\label{P_NL}
\end{equation}
and $B_{\Lambda,\beta}(n)$ will be defined in \eqref{cluster-coeff} and is such that $B_{\Lambda,\beta}(n)\to\beta_n$ as $\Lambda\to\mathbb{Z}^d$. 

Let $\mathcal{P}_{n+1,|\Lambda|}(\rho)$ be a polynomial of degree n+1 in $\rho$ defined as follows:
\begin{equation}
	\mathcal{P}_{n+1,|\Lambda|}(\rho):=\begin{cases}
		\rho\left(\rho-\frac{1}{|\Lambda|}\right)\cdots\left(\rho-\frac{n}{|\Lambda|}\right)\;\;\;{\rm{if}}\;\frac{n}{|\Lambda|}\le\rho,\\	
		0\;\;\;{\rm{otherwise}}.
	\end{cases}	
	\label{P_R-P_NL}
\end{equation}
Let us note that, when $\rho=\rho_{\Lambda}:=N/|\Lambda|$ we get
\begin{equation}
	\mathcal{P}_{n+1,|\Lambda|}(\rho_{\Lambda})=\frac{N}{|\Lambda|}P_{N,|\Lambda|}(n).
\end{equation}
Then, defining
\begin{equation} 
	\F(\rho):=\frac{1}{\beta}\left\{\rho(\log\rho-1)-\sum_{n\ge1}\frac{1}{n+1}\mathcal{P}_{n+1,|\Lambda|}(\rho)B_{\Lambda,\beta}(n) \right\}
	\label{ModFreeE}
\end{equation}
with $\rho\in[0,\mathcal{R}_C)$, and using Stirling's approximation we get:
\begin{equation}
	f_{\Lambda,\beta,\z}(N)=\F(\rho_{\Lambda})+S_{|\Lambda|}(\rho_{\Lambda}),
	\label{FeRho}
\end{equation} 
where $S_{|\Lambda|}(\rho_{\Lambda})$ is an error term of order $\log{\sqrt{|\Lambda|}}/|\Lambda|$ (see Appendix B of \cite{scola2020local}). Hence, 
\begin{eqnarray}
	f_{\beta}(\rho)=\lim_{\substack{\Lambda\to\mathbb{Z}^d\\N/|\Lambda|\to\rho}}f_{\Lambda,\beta,\z}(N)&=&\lim_{\substack{\Lambda\to\mathbb{Z}^d\\\rho_{\Lambda}\to\rho}}\F(\rho_{\Lambda})+S_{|\Lambda|}(\rho_{\Lambda})\nonumber\\
	&=&\rho(\log\rho-1)-\sum_{n\ge1}\frac{\beta_n}{n+1}\rho^{n+1}.
	\label{TFE_CE}
\end{eqnarray}

Furthermore, thanks to the validity of the cluster expansion we analyze the behavior of the truncated 2-point (\textquotedblleft canonical'') correlation function. Given $q_1,q_2\in\Lambda$ we define:
\begin{equation}
	u^{(2)}_{\Lambda,N}(q_1,q_2):=\rho^{(2)}_{\Lambda,N}(q_1,q_2)-\rho^{(1)}_{\Lambda,N}(q_1)\rho^{(1)}_{\Lambda,N}(q_2),
	\label{U2}
\end{equation}
where
\begin{equation}
	\rho^{(1)}_{\Lambda,N}(q):=\frac{1}{(N-1)!}\sum_{\mathbf{x}\in\Lambda^{N-1}}\frac{1}{Z^{per}_{\Lambda,\beta}(N)}e^{-\beta H^{per}_{\Lambda}(q,\mathbf{x})},
	\label{R1}
\end{equation}
and
\begin{equation}
	\rho^{(2)}_{\Lambda,N}(q_1,q_2):=\frac{1}{(N-2)!}\sum_{\mathbf{x}\in\Lambda^{N-2}}\frac{1}{Z^{per}_{\Lambda,\beta}(N)}e^{-\beta H^{per}_{\Lambda}(q_1,q_2,\mathbf{x})}
	\label{R2}
\end{equation}
where $Z^{per}_{\Lambda,\beta}(N)$ is given by \eqref{CanPF} with the difference that, for simplicity, we consider here periodic boundary conditions. We have: 
\begin{theorem} Let $q_1,q_2$ be two fixed points in the domain $\Lambda$, then there exist positive constants $C$ and $C_1$, independent of $N$ and $\Lambda$, such that, when $N/|\Lambda|$ is small enough, we have

	\begin{eqnarray}	
		|u^{(2)}(q_1,q_2)|&\le&\left(\frac{N}{|\Lambda|}\right)^2\bigg[(e^{4\beta J}-1)\mathbf{1}_{\{|q_1-q_2|=1\}}+\mathbf{1}_{\{q_1=q_2\}}\nonumber
		\\
		&+&\frac{(e^{4\beta J}-1)\mathbf{1}_{\{|q_1-q_2|=1\}}+\mathbf{1}_{\{q_1=q_2\}}}{N}+C e^{-|q_1-q_2|}\bigg]+ C_1\frac{1}{|\Lambda|}.\;\;\;\;\;\;\;\;\;\; 
		\label{2Point}
	\end{eqnarray}
	\label{TH2}
\end{theorem} 
The proof of the theorem will be given in Section \ref{SectionD}.

\subsection{Application of Theorem \ref{TH1} to Precise Large and Local Moderate Deviations Theorems} \label{SubSec3}

In the last part of the paper, using the validity of the cluster expansion in the canonical ensemble (Theorem \ref{TH1}),
we prove theorems on local moderate and precise large deviations. For that we follow the strategy in \cite{scola2020local} since the calculations done in $\mathbb{R}^d$ can be applied in $\mathbb{Z}^d$.
However, we present the particular case of $\mathbb{Z}^d$ because we want to compare with analogous results presented in \cite{del1974local} and \cite{dobrushin1994large} but using a different expansion. This comparison will be done in Section \ref{SectionLMD}.

%with respect to the {\it{grand-canonical probability measure}} 
%Then we compare with the 
%For that, we follow the strategy presented in \cite{scola2020local} for the case $\Lambda\subset\mathbb{R}^d$.
%{\color{red} 
%	{\color{blue}---------Motivation for this section, see point 2 of one of the reviewers---------}
%	  
%	The proofs of the main results presented here does not affected by the fact that we are considering here $\mathbb{Z}^d$ instead of $\mathbb{R}^d$, i.e. follows from Theorem \ref{TH1} and \cite{scola2020local}. But, in order to complete the exposition in Appendix \ref{S2} we will give a brief proof of the main theorems based on some technical lemmas given in Section 5 of 1\cite{scola2020local}.  What we want to emphasize here is the comparison with previous known an important results obtained in the Ising model and lattice gas system model, in particular with \cite{del1974local}, \cite{dobrushin1994large}.}

Fixing a {\it{chemical potential}} $\mu_0$ the {\it{grand-canonical probability measure}} at finite volume with zero boundary condition is defined as follows:
\begin{equation}
	\mathbb{P}^{\z}_{\Lambda,\mu_0}(\mathbf{x}):=\bigotimes_{N\ge0}\frac{1}{\Xi^{\z}_{\Lambda,\beta}(\mu_0)}\frac{e^{\beta \mu_0N}e^{-\beta H^{\z}_{\Lambda}(\mathbf{x})}}{N!}.
	\label{Prob1}
\end{equation}
where $\Xi^{\z}_{\Lambda,\beta}(\mu_0)$ is the {\it{grand-canonical partition function}} with zero boundary condition given by:
\begin{equation}
	\Xi^{\z}_{\Lambda,\beta}(\mu_0):=\sum_{N\ge 0}e^{\beta\mu_0 N}Z^{\z}_{\Lambda,\beta}(N),
	\label{GcPF1}
\end{equation}  
and where we used \eqref{CanPF}.

%\begin{eqnarray}
%\mathbb{P}^{\z}_{\Lambda,\mu_0}(\mathbf{x})&:=&\bigotimes_{N\ge0}\frac{e^{-\beta|\Lambda|\left(J\frac{|\mathcal{E}^{\mathbf{-1}}_{\Lambda}|}{|\Lambda|}+\frac{\mu_0}{2}\right)}\left(e^{-\beta \mu_0N}e^{-\beta H^{\z}_{\Lambda}(\mathbf{x})}\right)}{\tilde{\Xi}^{\mathbf{-1}}_{\Lambda,\beta}(h_{\Lambda})}\nonumber
%\\
%&=&\bigotimes_{N\ge0}\frac{e^{-\beta \mu_0N}e^{-\beta H^{\z}_{\Lambda}(\mathbf{x})}}{\Xi^{\z}_{\Lambda,\beta}(\mu_0)}.
%\end{eqnarray} 

Hence, using the results presented in Theorem \ref{TH1} we can study directly the probability of the set
\begin{equation}
	A_N:=\{\mathbf{x}\equiv\{x_i\}_{i\ge1},\;x_i\in\mathbb{Z}^d\;|\;|\mathbf{x}\cap\Lambda|=N\},
	\label{DeviationSet}
\end{equation}
for $N$ taking the value
$\tN$ being a general deviation from the mean value $\bN$ of order $\alpha\in[1/2,1]$, i.e,
\begin{equation}
	\tN:=\bN+u|\Lambda|^{\alpha},\;\;\;\alpha\in[1/2,1],\;u\in\mathbb{R}^+
	\label{Deviation}
\end{equation}
where
\begin{equation}
	\bRL:=\mathbb{E}^{\z}_{\Lambda,\mu_0}\left[\frac{N}{|\Lambda|}\right]=p'_{\Lambda,\beta,\z}(\mu_0),\;\;\;\bN:=\lfloor\bRL|\Lambda|\rfloor.
	\label{MeanV}
\end{equation} 

We have that
\begin{equation}
	\mathbb{P}_{\Lambda,\mu_0}^{\z}(A_{\tN})=\frac{e^{\beta\mu_0\tN}\canBC(\tN)}{\GcanBC(\mu_0)},
	\label{ProbA}
\end{equation}
and the key point is that we can now compute it using Theorem~\ref{TH1}.
For that purpose, we also define $N^*$ as the number of particles such that 
\begin{equation}
	\sup_{N}\left\{e^{\beta\mu_0N}\canBC(N)\right\}=e^{\beta\mu_0N^*}\canBC(N^*).
	\label{N*} 
\end{equation}

The previous number of particles is the one which allows to express the chemical potential $\mu_0$ using quantities on the canonical ensemble. Indeed $N^*$ has the following properties \cite{scola2020local}:   
\begin{equation}
	|\bN-N^*|\le C,
	\label{bN-N*}
\end{equation} 
with $C$ positive constant independent on $\Lambda$ and
\begin{equation}
	\mu_0=\F'(\rho^*_{\Lambda})+S'_{|\Lambda|}(\rho^*_{\Lambda}),
	\label{Mu0Rho*}
\end{equation}
where  $\rho^*_{\Lambda}=N^*/|\Lambda|$ and $S'_{|\Lambda|}(\cdot)$ has order $|\Lambda|^{-1}$ for all $\rho_{\Lambda}= N/|\Lambda|$.  With the notation $\F'(\cdot),$ $S'_{|\Lambda|}(\cdot)$ we mean the first derivative of $\F(\cdot),\;S_{|\Lambda|}(\cdot)$.

Note that from \eqref{bN-N*} we can rewrite $\tN$ given by \eqref{Deviation} as 
\begin{equation}
	\tN=N^*+u'|\Lambda|^{\alpha},
	\label{Deviation1}
\end{equation}
where $u'\sim u$.  With the notation $\sim$ we mean ``asymptotically'' as $|\Lambda|\to\infty$, where given two sequences $a_n\sim b_n\Leftrightarrow\lim_{n\to\infty}\frac{a_n}{b_n}=1$.

We have the following results.

\begin{theorem}[Precise Large Deviations]
	\label{TH-LD}
	Let $\mu_0\in\mathbb{R}$ be a chemical potential and let us fix zero boundary conditions. Let $\tilde{N}$ be a fluctuation given by \eqref{Deviation} with $\alpha=1$ such that Theorem  \ref{TH1} holds. Moreover let $\tilde{\mu}_{\Lambda}\in\mathbb{R}$ be the chemical potential such that,  $\tRL=\tN/|\Lambda|=:\mathbb{E}^{\z}_{\Lambda,\tilde{\mu}_{\Lambda}}[N/|\Lambda|]$. We have: 
	\begin{equation}
		\left|\probBC\left(A_{\tilde{N}}\right)-\frac{e^{-|\Lambda|I^{GC}_{\Lambda,\beta,\z}\left(\tRL;\bRL\right)}}{\sqrt{2\pi D_{\Lambda,\z}(\tRL^*)|\Lambda|}}\right|\le 
		\frac{Ce^{-|\Lambda|I^{GC}_{\Lambda,\beta,\z}\left(\tRL;\bRL\right)}}{|\Lambda|},
		\label{P3}
	\end{equation}
	where 
	\begin{equation}
		I^{GC}_{\Lambda,\beta,\z}\left(\tRL;\bRL\right):=\beta\left[\fgc(\tRL)-\fgc(\bRL)-\mu_0\left(\tRL-\bRL\right)\right]
		\label{LDO}
	\end{equation}
	and
	\begin{equation}
		D_{\Lambda,\z}(\tilde{\rho}^*_{\Lambda}):=\left[\beta \F''(\tilde{\rho}^*_{\Lambda})\right]^{-1}.
		\label{VAR}
	\end{equation}
	Here, with $\mathcal{F}''_{\Lambda,\beta,\z}(\cdot)$ we mean the second derivative of $\mathcal{F}_{\Lambda,\beta,\z}(\cdot)$, $\tN^*$ satisfies \eqref{N*} with $\tilde{\mu}_{\Lambda}$ instead of $\mu_0$, $\tRL^*=\tN^*/|\Lambda|$ and $\fgc(\cdot)$ is the {\it{grand-canonical free energy}} given by
	\begin{equation}
		\beta \fgc(\rho_{\Lambda}):=\sup_{\mu}\left\{\beta \rho_{\Lambda}\mu-\beta p_{\Lambda,\beta,\z}(\mu) \right\}.
		\label{GcFreeE} 
	\end{equation}
\end{theorem}	
In \eqref{GcFreeE} we used the {\it{finite volume pressure}} defined as follows:
\begin{equation}\label{FVLattice}
	\beta p_{\Lambda,\beta,\z}(\mu):=	\frac{1}{|\Lambda|}\log \Xi^{\z}_{\Lambda,\beta}(\mu),
\end{equation}	
where the grand-canonical partition function $\Xi^{\z}_{\Lambda,\beta}(\mu)$ is given by \eqref{GcPF1}.

\begin{remark} \label{RemarkF_GC}
	Note that: $\beta \fgc(\bRL)=\beta \bRL\mu_0-\beta p_{\Lambda,\beta,\z}(\mu_0)$, $(\fgc)'(\bRL)=\mu_0$ and $\beta \fgc(\tRL)=\beta \tRL\tilde{\mu}_{\Lambda}-\beta p_{\Lambda,\beta,\z}(\tilde{\mu}_{\Lambda})$.  
\end{remark}

\begin{theorem}[Local Moderate Deviations] 	\label{Th2}
	\label{TH-MD}
	Let $\mu_0\in\mathbb{R}$ be a chemical potential and let us fix zero boundary conditions. Let $N^*$ be as in \eqref{N*} such that Theorem \ref{TH1} holds and let us call $\rho^*_{\Lambda}=N^*/|\Lambda|$.
	For $\tilde N$ and the set $A_{\tilde N}$ respectively given by \eqref{Deviation1} and
	\eqref{DeviationSet} with $\alpha\in[1/2,1)$, we have:
	\begin{equation}\label{formulamoderate}
		\left|\probBC(A_{\tN})-\frac{\exp\left\{-\frac{(u')^2|\Lambda|^{2\alpha-1}}{2D^{\alpha}_{\Lambda,\z}(\rho^*_{\Lambda})}\right\}}{\sqrt{2\pi D^{\alpha,+}_{\Lambda,\z}(\rho^*_{\Lambda})|\Lambda|}}\right|\le \frac{2e^{-\frac{(u')^2|\Lambda|^{2\alpha-1}}{2D^{\alpha}_{\Lambda,\z}(\rho^*_{\Lambda})}}E_{|\Lambda|}(\alpha,u',\rho^*_{\Lambda})}{\sqrt{2\pi D^{\alpha,+}_{\Lambda,\z}(\rho^*_{\Lambda})|\Lambda|}},
	\end{equation}
	where
	\begin{equation}
		D^{\alpha}_{\Lambda, \z}(\rho^*_{\Lambda}):=\left[\beta\F''(\rho^*_{\Lambda})+\beta\sum_{m=3}^{m(\alpha)-1}\frac{2(u')^{m-2}\F^{(m)}(\rho^*_{\Lambda})}{m!|\Lambda|^{(m-2)(1-\alpha)}}\right]^{-1}
		\label{Var2}
	\end{equation}
	and
	\begin{equation}
		D^{\alpha,+}_{\Lambda, \z}(\rho^*_{\Lambda}):=\left[\beta\F''(\rho^*_{\Lambda})+\beta\sum_{m=3}^{m(\alpha)-1}\frac{2(u')^{m-2}|\F^{(m)}(\rho^*_{\Lambda})|}{m!|\Lambda|^{(m-2)(1-\alpha)}}\right]^{-1}.
		\label{Var1}
	\end{equation}
	Here, $m(\alpha)$ is given by $m(\alpha):=\min\left\{m\in\mathbb{N}\;|\;m(1-\alpha)-1>0\right\}$ and $E_{|\Lambda|}(\alpha,u',\rho^*_{\Lambda})$ is an error term of order $|\Lambda|^{-[(m(\alpha)(1-\alpha)-1]}$ defined via cluster expansion as
	\begin{equation}\begin{split}
			E_{|\Lambda|}(\alpha,u',\rho^*_{\Lambda}):=\frac{\beta}{|\Lambda|^{m(\alpha)(1-\alpha)-1}}\left|\frac{(u')^{m(\alpha)}\F^{(m(\alpha))}(\rho^*_{\Lambda})}{m(\alpha)!}+\frac{u'(\mu_0-\F'(\rho^*_{\Lambda})) }{|\Lambda|^{1-m(\alpha)(1-\alpha)-\alpha}}\right.
			\\
			\left.+\sum_{m\ge m(\alpha)+1}\frac{(u')^m\F^{(m)}(\rho^*_{\Lambda})}{m!|\Lambda|^{(m-m(\alpha))(1-\alpha)}}\right|,
		\end{split}
		\label{Error}
	\end{equation}
	where $\F^{(m)}(\cdot)$ is the $m$-th derivative of $\F(\cdot)$.
\end{theorem}

\begin{cor}[Local Central Limit Theorem.]	\label{CLT}
	Under the same assumptions as in Theorem \ref{Th2} for $\alpha=1/2$ we have that
	\begin{equation}\label{Cor1}
		\left|\probBC(A_{\tN})-\frac{\exp\left\{-\frac{(u')^2}{2 D_{\Lambda,\z}(\rho^*_{\Lambda})}\right\}}{\sqrt{2\pi D_{\Lambda,\z}(\rho^*_{\Lambda})|\Lambda|}}\right|\le \frac{2 e^{-\frac{(u')^2}{2 D_{\Lambda,\z}(\rho^*_{\Lambda})}}E_{|\Lambda|}(1/2,u',\rho^*_{\Lambda})}{\sqrt{2\pi D_{\Lambda,\z}(\rho^*_{\Lambda})|\Lambda|}},
	\end{equation}
	where, using \eqref{VAR}, 
	\begin{equation}
		D_{\Lambda,\z}(\rho^*_{\Lambda})=\left[\beta\F''(\rho^*_{\Lambda})\right]^{-1}
	\end{equation}
	and $E_{|\Lambda|}(1/2,u',\rho^*_{\Lambda})$ is an error term of order $|\Lambda|^{-1/2}$ defined via cluster expansion and given by \eqref{Error}. 
\end{cor}

For the proofs and the discussion related to the previous results we  refer to Section \ref{SectionLMD} and Appendix \ref{S2}.

\begin{remark}
	As we will see in Section \ref{SectionLMD}, thanks to Theorem \ref{TH1}, the main estimates \eqref{P3}, \eqref{formulamoderate} and \eqref{Cor1} can be obtained via a direct and explicit calculation. Moreover, the main quantities involved - \eqref{LDO}, \eqref{VAR}, \eqref{Var2},  \eqref{Var1}, \eqref{Error} - have also an explicit form in terms of sums of clusters.
\end{remark}

\section{Cluster Expansion and its convergence, proof of Theorem \ref{TH1}}
\label{SectionCE}

The proof follows closely the strategy in \cite{pulvirenti2012cluster}. 
For completeness of the presentation we repeat the main steps keeping track of the main modifications due to the lattice.
The key idea is to view the canonical partition function \eqref{CanPF} as a perturbation around the ideal case. Renormalizing with $|\Lambda|^N$ and defining  
\begin{equation}
	Z_{\Lambda,N}^{ideal}:=\frac{|\Lambda|^N}{N!}\;\;\mathrm{and}\;\;Z_{\Lambda,\beta,\z}^{int}(N):=\frac{1}{|\Lambda|^N}\sum_{\mathbf{x}\in\Lambda^N}e^{-\beta H_{\Lambda}^{\z}(\mathbf{x})},
	\label{Rewriting}
\end{equation}
we rewrite \eqref{CanPF} as
\begin{equation}
	Z^{\z}_{\Lambda,\beta}(N)=Z_{\Lambda,N}^{ideal}Z_{\Lambda,\beta,\z}^{int}(N).\nonumber
\end{equation}

Calling now $\mathcal{E}(N):=\{\{i,j\}\;|\;i,j\in\{1,...,N\}\}$ and $f_{i,j}\equiv f(x_i,x_j):=e^{-\beta V(x_i,x_j)}-1$
we have that the factor $e^{-\beta H_{\Lambda}^{\z}(\mathbf{x})}$ can be expressed as
\begin{equation}
	e^{-\beta H_{\Lambda}^{\z}(\mathbf{x})}=\prod_{1\le i<j\le N}(f_{i,j}+1)=\sum_{E\subset\mathcal{E}(N)}\prod_{\{i,j\}\in E}f_{i,j}\nonumber
\end{equation}
where the term +1 is given by  $E=\emptyset\subset\mathcal{E}(N)$. Note that we can associate to any set $E\in\mathcal{E}(N)$ a graph $g\equiv(V(g),E)$ where $V(g):=\{i\in\{1,...,N\}\;|\;\exists\;e\in E\;\mathrm{s.t}\;i\in e\}$ is the set of its vertices and $E$ is the set of its edges. Moreover, a graph created from $E$ does not contain isolated vertices and can be viewed as the pairwise compatible (non-ordered) collection of its connected components, where two graphs $g,\;g'$ are called \textit{compatible} ($g\sim g'$) if and only if $V(g)\cap V(g')=\emptyset$. In other terms given $E$ we can find a graph $g$ such that $g\equiv\{g_1,...,g_k\}_{\sim}$ with $k\ge1$, where, denoting by $\mathcal{C}_m$ the set of connected graphs with $m$ vertices, $g_l\in\mathcal{C}_m$  for all $l=1,...,k$ and $2\le m\le N.$ In this way we have that 
\begin{equation}\label{exp-H}
	e^{-\beta H_{\Lambda}^{\z}(\mathbf{x})}=\sum_{\substack{\{g_1,...g_k\}_{\sim}\\g_l\;\mathrm{connected}\;\forall\;l}}\prod_{l=1}^k\prod_{\{i,j\}\in E(g_l)} f_{i,j},%=\sum_{\substack{\{V_1,...V_k\}_{\sim}\\|V_l|\ge2}}\prod_{l=1}^k\zeta_{\Lambda}(V_l),\nonumber
\end{equation}
where the collection $\{g_1,...,g_k\}_{\sim}=\emptyset$ gives the term +1 in the sum. In  the sum in the right hand side of \eqref{exp-H} we imply the presence of a sum over $k\ge1$, in the sense that it runs over all collections of $k\ge1$ connected compatible graphs with sets of vertices in $\{1,...,N\}$.

Hence, denoting with  $\mathcal{C}_V$ the set of connected graphs with set of vertices $V$ and defining 
\begin{equation}
	\zeta_{\Lambda}(V):=\sum_{g\in\mathcal{C}_V}\frac{1}{|\Lambda|^{|V|}}\sum_{\mathbf{x}\in\Lambda^{|V|}}\prod_{\{i,j\}\in E(g)}f_{i,j},
	\label{zeta1}
\end{equation}
we get
\begin{eqnarray}
	Z^{int}_{\Lambda,\beta,\z}(N)=\sum_{\substack{\{V_1,...,V_k\}_{\sim}\\|V_l|\ge 2\;\forall\;l}}\prod_{l=1}^k\zeta_{\Lambda}(V_l)=\exp\left\{\sum_{I\in\mathcal{I}}c_I\zeta^I_{\Lambda}\right\},
	\label{CpfCluster}
\end{eqnarray}
where $V_l\equiv V(g_l),\;l=1,...,k$ and the second equality of \eqref{CpfCluster} holds under the validity of Lemma \ref{LemmaCluster} and where we used 
\begin{equation}
	c_I := \frac{1}{I!} \sum_{G\in\mathcal{G}_I}(-1)^{|E(G)|}
	\label{cI}
\end{equation}
which comes from the polymer model representation described below (see also Section 3 of \cite{pulvirenti2012cluster}).  

An {\it{abstract polymer model}} $(\Delta,\mathbb{G}_{\Delta} ,\omega)$ consists of a set of polymers $\Delta := \{\delta_1, . . . , \delta_{|\Delta|}\}$, a compatibility graph $\mathbb{G}_{\Delta}$ with set of vertices $\Delta$ and set of edges $E_{\Delta}$ such that $\{i,j\}\in E_{\Delta}$ if and only if $\delta_i\not\sim \delta_j$ (i.e. $\delta_i\cap\delta_j\ne\emptyset$) and a weight function $\omega : \Delta \to \mathbb{R}$. 

In our case the set of polymers is given by $\mathcal{V} := \{V\subset\{1,...,N\},\;|V|\ge2\}$
and the weight function is $\zeta_{\Lambda}(V)$ defined in \eqref{zeta1}. In the second equality in \eqref{CpfCluster} the sum in the exponent  is over the set $\mathcal{I}$ of all multi-indices $I : \mathcal{V} \to \{0,1,...\}$, with $\zeta^I_{\Lambda} =\prod_{V} \zeta(V)^{I(V)}$. Denoting also with $\mathrm{suppI} := \{V \in \mathcal{V}\;|\;I(V) > 0\},\;\mathcal{G}_I$ is the graph with
$\sum_{V\in\mathrm{suppI}}I(V)$ vertices induced from $\mathcal{G}_{\mathrm{suppI}} \subset\mathbb{G}_{\mathcal{V}}$ by replacing each vertex $V$ by the complete graph on $I(V)$ vertices. 

We recall that (as it is observed in \cite{pulvirenti2012cluster}), the sum in \eqref{cI} is over all connected subgraphs $G$ of $\mathcal{G}_I$ spanning the whole set of vertices of $\mathcal{G}_I$ and $I! =\prod_{V\in\mathrm{suppI}} I(V)!$, indeed if $I$ is  not a cluster (i.e. $\mathcal{G}_{\rm{suppI}}$ is not connected), then $c_I = 0$.

Then from Sections 5 and 6 of \cite{pulvirenti2012cluster} and using the representation above we have that $F_{\beta,N,\Lambda}(n)$ is given by \eqref{F_bNL} where now we can define rigorously $B_{\Lambda,\beta}(n)$ as follows:
\begin{equation}
	B_{\Lambda,\beta}(n):=\frac{|\Lambda|^n}{n!}\sum_{A(I)=[n+1]}c_I\zeta_{\Lambda}^I,
	\label{cluster-coeff}
\end{equation}
where $A(I) := \bigcup_{V\in\mathrm{supp}\;I} V \subset\{1,..,N\}$ and $[n+1]:=\{1,...,n+1\}$.

The convergence of the cluster expansion is guaranteed from the following Lemma, in which we follow Theorem 1 - (ii) in \cite{morais2013continuous} and we use the tree graph inequality as it is presented in \cite{procacci2017convergence}. 

\begin{lemma}\label{LemmaCluster}
	There exist  constants $\mathcal{R}^C$ and $a>0$, such that when  $N/|\Lambda|<\mathcal{R}^C$, the following holds:
	\begin{equation}
		\sup_{i\in\{1,..,N\}}\sum_{V\in\mathcal{V}(N)\;:\;i\in V}|\zeta_{\Lambda}(V)|e^{a|V|}\le e^a-1.
		\label{E1L1}	
	\end{equation}
\end{lemma}

\begin{proof} 
	
	We start, noting that Proposition 1 in \cite{procacci2017convergence} is here valid, i.e.,  we have the validity of the following tree-graph inequality:
	\begin{equation}
		\left|\sum_{g\in\mathcal{C}_n}\prod_{\{i,j\}\in E(g)}f_{i,j}\right|\le e^{\beta Bn}\sum_{T\in\mathcal{T}_n}\prod_{\{i,j\}\in E(T)}(1-e^{-\beta |V(x_i-x_j)|}),
		\label{Prop1-PY}
	\end{equation}
	where $\mathcal{T}_n$ is the  set of  trees with $n=|V|$ vertices.
	Then, given a rooted tree $T\in\mathcal{T}_n$ with set of edges given by  $E(T)=\{(i_1,j_1),...,(i_{n-1},j_{n-1})\}$  and defining $d_1(x,\Lambda^c):=\inf_{x'\in\Lambda^c}\{|x-x'|\}$ with $x\in\mathbb{Z}^d$, we get:
	
	\begin{eqnarray}
		&&\sum_{\mathbf{x}\in\Lambda^n}\frac{1}{|\Lambda|^n}\prod_{\{i,j\}\in E(T)}(1-e^{-\beta|V(x_i-x_j)|}) \le  \frac{1}{|\Lambda|^n}\sum_{x_{i_1}\in\Lambda} \sum_{\mathbf{y}\in\Lambda^{n-1}}\prod_{k=2}^{n}(1-e^{-\beta |V(y_k)|})\label{E3}\;\;\;\;\;\;\;\;\\
		&&\;\;\;\;\;\;\;\;=\frac{1}{|\Lambda|^n}\sum_{x_{i_1}\in\Lambda} \sum_{\mathbf{y}\in\Lambda^{n-1}}\prod_{k=2}^{n}(1-e^{-\beta |V(y_k)|})(\mathbf{1}_{\{d_1(x_{i_1},\Lambda^c)\le1\}}+\mathbf{1}_{\{d_1(x_{i_1},\Lambda^c)>1\}})\nonumber\\
		&&\;\;\;\;\;\;\;\;\;\;\;\;\;\;\;\;\;\;\;\;\;\;\;\;\;\;\;\;\;\;\;\;\;\;\;\; \le\frac{[\bar{C}_{J,d}(\beta)]^{n-1}}{|\Lambda|^n}\sum_{x_{i_1}\in\Lambda}(\mathbf{1}_{\{d_1(x_{i_1},\Lambda^c)\le1\}}+\mathbf{1}_{\{d_1(x_{i_1},\Lambda^c)>1\}})\nonumber\\
		&&\;\;\;\;\;\;\;\;\;\;\;\;\;\;\;\;\;\;\;\;\;\;\;\;\;\;\;\;\;\;\;\;\;\;\;\;\;\;\;\;\;\;\;\;\;\;\;\;\;\;\;\;\le\begin{cases}
			\frac{1}{|\Lambda|^{n-1}}[\bar{C}_{J,d}(\beta)]^{n-1},\;{\rm{if}}\;d_1(x_{i_1},\Lambda^c)>1\\
			\frac{|\partial\Lambda|}{|\Lambda|^{n}}[\bar{C}_{J,d}(\beta)]^{n-1},\;\;\;\;\;{\rm{if}}\;d_1(x_{i_1},\Lambda^c)\le1.\\
		\end{cases}\nonumber
	\end{eqnarray}
	In \eqref{E3} we considered $i_1$ as a root, $\mathbf{y}$ is a vector in $\Lambda^{n-1}$ with components $y_k=x_{i_k}-x_{j_k},\;\forall\,k=2,...,n$, and fixing $x^*\in\mathbb{Z}^d$,  we defined:
	\begin{equation}
		\bar{C}_{J,d}(\beta):=\sum_{x\in\mathbb{Z}^d}(1-e^{-\beta\left|V(x-x^*)\right|})=1+2d(1-e^{-4\beta J}).
		\label{C-Bar}
	\end{equation} 
	Let us note that $|\partial\Lambda|/|\Lambda|$  vanishes as $\Lambda \to \mathbb{Z}^d$ for a suitable sequence $(\Lambda_n)_{n\ge 1}$. For the analogous calculation in the continuous case we also refer to  \cite{pulvirenti2015finite} formulas (4.18)-(4.21).
	
	%{\color{blue}
	%	
	%	old version: In the previous estimate the boundary term gives a similar contribution which, however, is multiplied for  $|\partial\Lambda|/|\Lambda|$ as in formulas (4.18)-(4.21) in \cite{pulvirenti2015finite}.
	%
	%new possible version: For completeness, we recall that, following \cite{pulvirenti2015finite} - formulas (4.18)-(4.21) - we also have the following:
	%\begin{eqnarray}
	%	&&\sum_{\mathbf{x}\in\Lambda^n}\frac{1}{|\Lambda|^n}\prod_{\{i,j\}\in E(T)}(1-e^{-\beta|V(x_i-x_j)|}) \le  \frac{1}{|\Lambda|^n}\sum_{x_{i_1}} \sum_{\mathbf{y}\in\Lambda^{n-1}}\prod_{k=2}^{n}|1-e^{-\beta V(y_k)}| \nonumber
	%	\\
	%	&&\;\;\;\;\le\frac{|\partial\Lambda|}{|\Lambda|}
	%	\frac{1}{|\Lambda|^{n-1}}[\bar{C}_{J,d}(\beta)]^{n-1}\le 	\frac{1}{|\Lambda|^{n-1}}[\bar{C}_{J,d}(\beta)]^{n-1}.
	%	\label{BoundaryTerm}
	%\end{eqnarray}}
	Hence, from \eqref{Prop1-PY} and \eqref{E3} we can write:
	\begin{equation}
		|\zeta_{\Lambda}(V)|\le\frac{n^{n-2}}{|\Lambda|^{n-1}}e^{\beta Bn}\,[\bar{C}_{J,d}(\beta)]^{n-1},
		\label{E4L1}
	\end{equation}
	where $B$ is the stability constant defined in \eqref{Stability}.
	
	Fixing now $i\in \{1,...,N\},$ and using the fact that $\zeta_{\Lambda}(V)$ depends only on $|V|$, from \eqref{E4L1}, for the left hand side of \eqref{E1L1} we get
	\begin{eqnarray}\label{DaRimettere}
		\sup_{i\in\{1,..,N\}}&&\sum_{V\in\mathcal{V}(N)\;:\;i\in V}|\zeta_{\Lambda}(V)|e^{a|V|}\le e^{a+\beta B}\times\\
		&&\times\sum_{n=2}^N {{N-1}\choose{n-1}} \frac{n^{n-2}}{|\Lambda|^{n-1}}\left[e^{(\beta B+a)}\bar{C}_{J,d}(\beta)\right]^{n-1}.\nonumber
	\end{eqnarray}
	The latter implies that \eqref{E1L1} is verified when the following is true
	\[\sum_{n\ge1}\left[\frac{N}{|\Lambda|}e^{(\beta B+a)}\bar{C}_{J,d}(\beta)\right]^{n-1}\frac{n^{n-1}}{n!}\le 1+e^{-\beta B}(1-e^{-a}).\]
	
	Then using the result from \cite{morais2013continuous} (formulas (3.12)-(3.15)) we have that the cluster expansion is absolutely convergent (uniformly in $N$ and $\Lambda$) when
	\begin{equation}
		\frac{N}{|\Lambda|}\le\mathcal{R}_C,\nonumber
		\label{ConvCan}
	\end{equation}
	where	
	\begin{equation}
		\mathcal{R}_C:=[e^{\beta B}\bar{C}_{J,d}(\beta)]^{-1}\left\{\max_{a>0}\frac{\ln[1+e^{-\beta B}(1-e^{-a})]}{e^a[1+e^{-\beta B}(1-e^{-a})]}\right\}.
		\label{RadCan}
	\end{equation}
	%and
	%	\begin{equation}
	%	\mathcal{K}(u):=\max_{a>0}\frac{\ln[1+u(1-e^{-a})]}{e^a[1+u(1-e^{-a})]}.\nonumber
	%	\end{equation} 
\end{proof}
For the conclusion of the proof of Theorem~\ref{TH1} we refer the reader to \cite{pulvirenti2012cluster}, Sections 5 and 6. 

\subsection{Some remarks}
\label{RemarKacBC}

Let us given some more precise examples how we can generalize our approach.

{\textbf{Kac potential.}} We consider first the Ising model with a Kac potential as it formalized in Section 4.2.1 and  Section 9 of  \cite{presutti2008scaling}. Moreover, we recall  that this kind of potential is the one considered also in \cite{farrell1966cluster}.

Hence, the Hamiltonian \eqref{HamIs} is here given by:

\begin{equation}
	\mathcal{H}^{\boldsymbol{\sigma}}_{\Lambda,R}(\boldsymbol{\sigma})	:=-\sum_{\{x,x'\}\in\mathcal{E}_{\Lambda,R}} J_R(|x-x'|)\sigma(x)\sigma(x'),
\end{equation}	
where, given $0<R<<L$ - with $\Lambda=(-L/2,L/2]^d\cap\mathbb{Z}^d,\;L\in\mathbb{Z}$ - we defined 
\[\mathcal{E}_{\Lambda,R}:=\{\{x,x'\}\in\mathbb{Z}^d\;|\;\{x,x'\}\cap\Lambda\ne \emptyset,\;0<R|x-x'|\le 1\}\]
and
\begin{equation}
	J_R(|x-x'|):=R^d J(|Rx-Rx'|)>0.
\end{equation}	 
The function $J(\cdot)$ satisfies the following assumptions: (i)  $J(|x'-y'|)=J(|x-y|)$ where for all $a\in\mathbb{R}^d,\;x'=x+a,\;y'=y+a$; (ii) $J(r)$ is a non-negative, $C^2$ function, supported by the unit ball and such that $\int_{\mathbb{R}^d}J(r)dr=1$.  See for example Figure 4.4 in \cite{presutti2008scaling}.
Then, passing to the lattice gas system, i.e., applying \eqref{transf}, we have
\begin{eqnarray}\label{Kac1}
	\mathcal{H}^{\boldsymbol{\eta}}_{\Lambda,R}(\boldsymbol{\eta})\equiv	\mathcal{H}^{\boldsymbol{\sigma}}_{\Lambda,R}(\boldsymbol{\sigma})	&=&4m'|J_R(\mathcal{E}_{\Lambda,R})|-|J_R(\mathcal{E}_{\Lambda,R})|\nonumber\\
	&-&4\sum_{\{x,x'\}\in\mathcal{E}_{\Lambda,R}} J_R(|x-x'|)\eta(x)\eta(x'),
\end{eqnarray}	
where $|J_R(\mathcal{E}_{\Lambda,R})|:=\sum_{\{x,x'\}\in\mathcal{E}_{\Lambda,R}}J_R(|x-x'|)$.

Equation \eqref{Kac1} implies that our potential will be given here by the following:
\begin{equation}
	V_R(x-x'):=\begin{cases}
		\infty,\;\;\;\;\;\;\;\;\; \;\;\;\;\;\; \;\;\;\;\;\; \;x=x',\\
		-4 J_R(|x-x'|),\;\;0<R|x-x'|\le 1,\\
		0,\;\;\;\;\;\;\;\;\;\;\;\;\;\; \;\;\;\;\;\; \;\;\; {\rm{otherwise}}.	
	\end{cases}		
\end{equation}	 
Thus, if we consider for instance 
\[J(|Rx-Rx'|):=\frac{\mathbf{1}_{\{|x-x'|\le R\}}}{R^d}\]we will find that our stability constant as well as the regularity are not more given by \eqref{Stability} and \eqref{TildeC}, but by 
\[B_R:=-8Rd\;\;\;{\rm{and}}\;\;\; C_{d,R}(\beta):=2dR(e^{4\beta}-1)+1.\]
Let also note that, instead of \eqref{C-Bar} we will find \[\bar{C}_{d,R}(\beta):=1+2dR(1-e^{-4\beta}).\]
Hence, having these quantities, rewriting properly \eqref{Prop1-PY}, \eqref{E3} and \eqref{RadCan}, the validity of Lemma \ref{LemmaCluster} and consequently Theorem \ref{TH1} is still true also in this case.

{\textbf{Non-zero boundary conditions.}} Lemma \ref{LemmaCluster} and then Theorem \ref{TH1} also hold true if we consider $\g\ne\z$ fixed boundary conditions.  Indeed, defining  $\nu_{\Lambda}(x_i|\g):=e^{-\beta \sum_{j\ge1}V(x_i-\gamma_j)}>0$, 
which is 1 if $d_1(x,\Lambda^c)>1$, we can write \eqref{CanPF} as
\begin{equation}
	Z_{\Lambda,\beta}^{\g}(N)=\frac{1}{N!}\sum_{\mathbf{x}\in\Lambda^N}e^{-\beta H^{\mathbf{0}}_{\Lambda}(\mathbf{x})}\prod_{i=1}^{N}\nu_{\Lambda}(x_i|\g),\nonumber
\end{equation}
where we used 
\begin{equation}
	H^{\g}_{\Lambda}(\mathbf{x})=H^{\z}_{\Lambda}(\mathbf{x})+\sum_{\substack{1\le i \le N, x_i\in\Lambda\\ j\ge1,\; \gamma_j\in\Lambda^c}} V(x_i-\gamma_j). \nonumber
\end{equation}
Then noting that 
\begin{equation}
	e^{\beta B}\le \nu_{\Lambda}(x_i|\g)\le e^{\beta d B}\nonumber%\;\;{\rm{and}}\;\;\prod_{i=1}^n\nu_{\Lambda}(x_i|\g)\le\prod_{i=1}^{|\partial\Lambda|}\nu_{\Lambda}(x_i|\g)\;\forall\;n\le N,\nonumber 
\end{equation}	
estimate \eqref{E3} is here given by 
\begin{eqnarray}
	&&\sum_{\mathbf{x}\in\Lambda^n}\prod_{i=1}^n\frac{\nu_{\Lambda}(x_i|\g)}{|\Lambda|}\prod_{(i,j)\in E(T)}|f_{i,j}|
	% &&\;\;\;\le \frac{1}{|\Lambda|^n} \sum_{\mathbf{x}\in\Lambda^n}\prod_{i=1}^{|\partial\Lambda|}\nu_{\Lambda}(x_i|\g)\prod_{k=1}^{n-1}|f_{i_k,j_k}||F(\epsilon)|\nonumber
	\le \begin{cases}
		\frac{C_{J,d}(\beta)^{n-1}}{|\Lambda|^{n-1}},\;\;\;\mathrm{if}\;d_1(x_i,\Lambda^c )>1\;\forall\;i=1,...,n,
		\\
		\\
		\frac{e^{\beta d B}}{|\Lambda|^{n-1}}\left\{\frac{|\partial\Lambda|}{|\Lambda|}[e^{\beta d B}C_{J,d}(\beta)]^{n-1}\right\},\;\;\mathrm{otherwise}.\nonumber
	\end{cases}
\end{eqnarray}

To obtain the contribution of order $|\partial\Lambda|/|\Lambda|$, we proceeded - and we can also  conclude - as for the case of zero boundary condition - see \eqref{E3}.

{\textbf{Penrose tree-graph inequality.}} The usual estimate done using the \textquotedblleft classical\textquotedblright\space tree graph inequality due to Penrose - see Theorem 0 in \cite{procacci2017convergence} - gives us the following (see also formula (3.10)-(3.11) in \cite{morais2013continuous}):
\begin{equation}
	|\zeta_{\Lambda}(V)|\le \frac{n^{n-2}}{|\Lambda|^{n-1}}e^{2\beta B(n-2)}[C_{J,d}(\beta)]^{n-1},
	%\label{E3}
\end{equation}
instead of  \eqref{E3}. We used $|V|=n$ and $C_{J,d}(\beta)$ given in \eqref{TildeC}, which is such that $\bar{C}_{J,d}(\beta)\le C_{J,d}(\beta)$.

Hence, instead of \eqref{DaRimettere} we find 
\begin{eqnarray}\label{DaRimettere1}
	\sup_{i\in\{1,..,N\}}&&\sum_{V\in\mathcal{V}(N)\;:\;i\in V}|\zeta_{\Lambda}(V)|e^{a|V|}\le e^{a-2\beta B}\times\nonumber\\
	&&\times\sum_{n=2}^N {{N-1}\choose{n-1}} \frac{n^{n-2}}{|\Lambda|^{n-1}}\left[e^{(2\beta B+a)}C_{J,d}(\beta)\right]^{n-1}.
\end{eqnarray}

In this way, applying directly  the (ii) of Theorem 1 in \cite{morais2013continuous}, instead of \eqref{RadCan} we find the following:
\begin{equation}\label{R-MP}
	\bar{\mathcal{R}}_C:=[e^{2\beta B}C_{J,d}(\beta)]^{-1}	\left\{\max_{a>0}\frac{\ln[1+e^{2\beta B}(1-e^{-a})]}{e^a[1+e^{2\beta B}(1-e^{-a})]}\right\}.
\end{equation}	
Let us also note that the new estimate given in Lemma \ref{LemmaCluster}, does not depend on the fact that we are on $\mathbb{Z}^d$. Hence, under a proper reformulation and assumptions, this could be applied also in the continuous case ($\Lambda\subset\mathbb{R}^d$).  

We want to briefly underline the following fact. Let us define 
\[\mathfrak{F}(u):=\max_{a>0}\frac{\ln[1+u(1-e^{-a})]}{e^{a}[1+u(1-e^{-a})]}.\]
In \cite{morais2013continuous},  the authors apply a more refined analysis for the estimate \eqref{E1L1}, in order to obtain a better convergence condition than the usual one (compare, for example, \eqref{R-MP} with the result in \cite{pulvirenti2012cluster}). Considering now $\mathcal{R}_C$ - \eqref{RadCan} - and $\bar{\mathcal{R}}_C$ define above, we can say what follows. On one hand, we have that for all $\beta,B,J>0,\;[e^{\beta B}\bar{C}_{J,d}(\beta)]^{-1}$ is bigger than $[e^{2\beta B}C_{J,d}(\beta)]^{-1}$. On the other, $\mathfrak{F}(e^{2\beta B})$ is bigger than $\mathfrak{F}(e^{-\beta B})$,when $\beta$ is \textquotedblleft small enough\textquotedblright - dependently on $d$ and $J$ -  in such a way that $\bar{\mathcal{R}}_C>\mathcal{R}_C$.
A comparison between $\mathcal{R}_C$ and $\bar{\mathcal{R}}_C$ is  given in Figure \ref{CanCan1} below, for $J=1,\;d=1,2,3$ and $\beta\in[0,1]$. As we can see, there exists inverse temperature $\beta^*\equiv\beta^*(J,d)$ such that  $\mathcal{R}_C\le \bar{\mathcal{R}}_C$ for $\beta\le\beta^*$, as well as $\bar{\mathcal{R}}_C< \mathcal{R}_C$ for $\beta> \beta^*$. It is also possible to recover the same behavior if we fix $d$, and we vary $J$. 

\begin{figure}[H]
	\centering
	\includegraphics[width=1
	\textwidth]{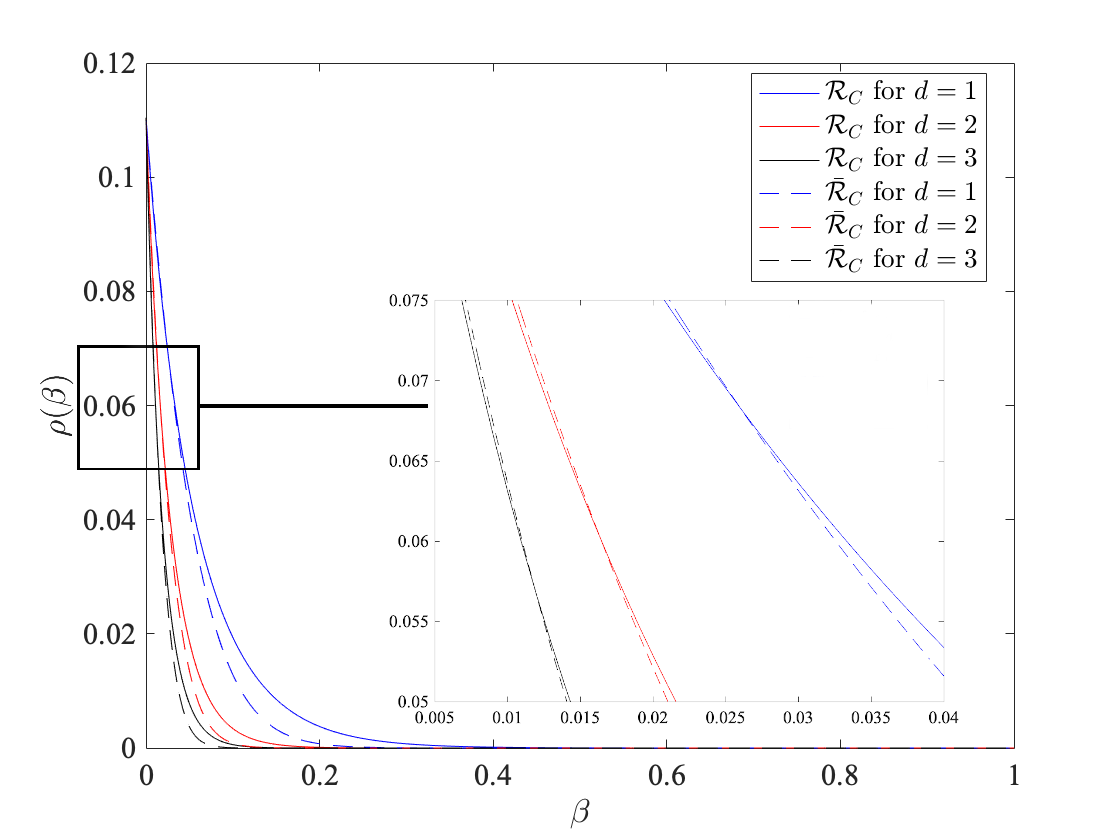}
	\caption{$\mathcal{R}_{C}$ (continuous line) and $\bar{\mathcal{R}}_{C}$ (dashed line) with $J=1$ and $\beta\in[0,1]$, in dimension 1 (blue lines), 2 (red lines) and 3 (black lines).}
	\label{CanCan1}
\end{figure}

\section{Grand canonical ensemble}\label{SectionRD}

In this section we will consider various representation of the grand-canonical descriptions and related results. In particular:  

\begin{enumerate}
	
	\item first, in Subsection \ref{SubSecGC0} , we relate the grand-canonical partition function for the Ising model with {\it{external magnetic field}} $h$, with the one for a lattice gas system with chemical potential $\mu$, as we already done in the canonical ensemble;
	
	\item second, in Subsection \ref{SubSecGC1}, we will establish the  condition of convergence for the cluster expansion for the Ising model with the contour representation as it is presented in Chapter 5 of \cite{friedli2017statistical}, and the one for the lattice gas system using the results presented in \cite{procacci2017convergence} and we compare them graphically;
	
	\item third, in Subsection \ref{SubSecGC2}, we find the lower bound of the density radius of convergence for the virial inversion of the lattice gas model and we compare it graphically with the value obtained in the canonical ensemble, given by \eqref{RadCan}.
\end{enumerate}

\subsection{Ising model and related lattice gas system in the grand-canonical ensemble.}
\label{SubSecGC0}

Using the Hamiltonian defined in \eqref{HamIs}, the {\it{grand-canonical partition function for the ferromagnetic  Ising model}} in a box $\Lambda\subset\mathbb{Z}^d$  with $-1$ boundary conditions is given by:
\begin{equation}
	\tilde{\Xi}^{-}_{\Lambda,\beta}(h):=\sum_{\sig\in\{-1,1\}^{\Lambda}}e^{\beta h\sum_{x\in\Lambda}\sigma(x)-\beta\mathcal{H}_{\Lambda}^{-}(\sig)}=\sum_{\substack{m\;:\\m|\Lambda|=\sum_{x\in\Lambda}\sigma(x)}}e^{\beta h m|\Lambda|}\tilde{Z}_{\Lambda,\beta}^{-}(m),
	\label{GC-Is}
\end{equation} 
where with the apex \textquotedblleft -\textquotedblright\space we mean  $\sigma^c=-1$, $\tilde{Z}_{\Lambda,\beta}^{-}(m)$ is given by \eqref{CpfIs}, $h$ is an external magnetic field and since we will work close to the -1 phase we will consider $h\le 0$. Using \eqref{GC-Is} we define the {\it{finite volume}} and {\it{thermodynamic pressure}} for the Ising model respectively as given by:
\begin{equation}\label{FVPressIs}
	\beta \psi_{\Lambda,\beta,-1}(h):=\frac{1}{|\Lambda|}\log \tilde{\Xi}^{\mathbf{-}}_{\Lambda,\beta}(h)
\end{equation}	
and
\begin{equation}\label{ThermodynPressIs}
	\psi_{\beta}(h):=\lim_{\Lambda\to\mathbb{Z}^d}\psi_{\Lambda,\beta,-1}(h).
\end{equation}

%	Similarly, we define the {\it{grand-canonical partition function}} for a lattice gas system in a box $\Lambda\subset\mathbb{Z}^d$ with zero boundary conditions, potential given by \eqref{Potential} and chemical potential $\mu \in \mathbb{R}$, in the following way:
%	\begin{equation}
%		\Xi^{\z}_{\Lambda,\beta}(\mu):=\sum_{N\ge 0}e^{\beta\mu N}Z^{\z}_{\Lambda,\beta}(N).
%		\label{GcPF1}
%	\end{equation}  

Moreover, we recall that using the partition function  $\Xi^{\z}_{\Lambda,\beta}(\mu)$ given by \eqref{GcPF1}, the finite volume pressure $p_{\Lambda,\beta,\z}(\mu)$ is given by \eqref{FVLattice} and hence, the {\it{thermodynamic pressure}} is defined as follows:
\begin{equation}
	p_{\beta}(\mu):=\lim_{\Lambda\to\mathbb{Z}^d}p_{\Lambda,\beta,\z}(\mu).
	\label{TdP}
\end{equation}

From \eqref{GcPF1}, using \eqref{transf} - \eqref{Cpf}, we have: 
\begin{equation}
	\Xi^{\z}_{\Lambda,\beta}(\mu)=\exp\left\{\beta|\Lambda|\left[\frac{\mu}{2}+J\frac{|\mathcal{E}_{\Lambda}|}{|\Lambda|}\right]\right\}\tilde{\Xi}^{-}_{\Lambda,\beta}(h_{\Lambda}),
	\label{GcPF2}
\end{equation}
where
\begin{equation}
	h_{\Lambda}\equiv h_{\Lambda}(\mu):=\frac{\mu}{2}+\frac{|\mathcal{E}_{\Lambda}|}{|\Lambda|}.
	\label{Mu-H}
\end{equation}

On the other hand, in a similar way, if we start form the Ising model, i.e., given $h\in\mathbb{R}$ we get 
\begin{equation}
	\tilde{\Xi}^{-}_{\Lambda,\beta}(h)=\exp\left\{-\beta|\Lambda|\left[h-J\frac{|\mathcal{E}_{\Lambda}|}{|\Lambda|}\right]\right\}\Xi^{\z}_{\Lambda,\beta}(\mu_{\Lambda}),
	\label{GcPF}
\end{equation}
where
\begin{equation}\label{MU-H}
	\mu_{\Lambda}\equiv\mu_{\Lambda}(h):=2h-4J\frac{|\mathcal{E}_{\Lambda}|}{|\Lambda|}.
\end{equation}

Furthermore, we recall that - in the framework considered here, i.e., far away from the phase transitions - between then thermodynamic free energy \eqref{FreeE} and the thermodynamic pressure \eqref{TdP}, the following Legendre transform relations occurs:  
\begin{equation}
	\beta f_{\beta}(\rho)=\sup_{\mu}\{\rho\mu-\beta p_{\beta}(\mu)\}
	\label{LeTr}
\end{equation}
and
\begin{equation}
	\beta p_{\beta}(\mu)=\sup_{\rho}\{\rho\mu-\beta f_{\beta}(\rho)\}.
	\label{LeTr1}
\end{equation}

\subsection{Cluster expansion of \eqref{GC-Is} and \eqref{GcPF1}.}\label{SubSecGC1}

For the cluster expansion of \eqref{GC-Is} we use the contour ensemble, i.e., the polymer model represention of $\tilde{\Xi}^-_{\Lambda,\beta}(h)$ as it is presented in Chapter 5 of \cite{friedli2017statistical}. Hence, we rewrite  \eqref{GC-Is} as  
\begin{equation}
	\tilde{\Xi}^{\mathbf{-}}_{\Lambda,\beta}(h)=\exp\left\{\beta|\Lambda|\left[J\frac{|\mathcal{E}_{\Lambda}|}{|\Lambda|}-h\right]\right\}\Xi^{Int}_{\Lambda,\beta}(z_h),\nonumber
	\label{Xi}
\end{equation}
where
\begin{equation}
	\Xi^{Int}_{\Lambda,\beta}(z_h):=1+\sum_{n\ge1}\frac{1}{n!}\sum_{S_1}\cdot\cdot\cdot\sum_{S_n}\prod_{1\le i<j\le n}(\hat{f}_{i,j}+1)\prod_{i=1}^nw(S_i)z_h^{|S_i|}.\nonumber
	\label{Xi-Int}	
\end{equation}
In the last definition we used the following objects: $S$ is a maximal connected subset of $\Lambda$ satisfying (i) $\sigma(x)=+1$ for all $x\in S$ and (ii) $|x-x'|=1$ for all $x,x'\in S$,

\begin{equation}
	\hat{f}_{i,j}\equiv \hat{f}(S_i,S_j):=\begin{cases}
		-1\;\;{\rm{if}}\;\inf\{|x-x'|\;x\in S_i,\;x'\in S_j\}\le1,\\
		0\;\;\;\;\,{\rm{otherwise}},\nonumber
	\end{cases}
\end{equation}
$z_h=\exp\{2\beta h\}$ and  $w(S):=\exp\left\{-2\beta J|\partial_eS|\right\},$ where $\partial_eS:=\{\{x,x'\}\;|\;|x-x'|=1,\;x\in S,\;x'\notin S\}.$

In this case, denoting with $[S]_1:=\{x\in\mathbb{Z}^d\;|\;d_1(x,S)\le1\}$, we have 
\begin{equation}
	\log\Xi^{Int}_{\Lambda,\beta}(h)=\sum_{n\ge 1}\sum_{S_1\subset\Lambda}\cdot\cdot\cdot\sum_{S_n\subset\Lambda}\frac{1}{n!}\sum_{g\in\mathcal{C}_n}\prod_{\{i,j\}\in E(g)}\hat{f}_{i,j}\prod_{i=1}^n\nonumber w(S_i)z_h^{|S_i|}	
	\label{ClusterContours}
\end{equation}
and
\begin{equation}
	1+\sum_{n\ge2}\frac{1}{(n-1)!}\sum_{S_2\subset\Lambda}\cdot\cdot\cdot\sum_{S_n\subset\Lambda}\left|\sum_{g\in\mathcal{C}_n}\prod_{\{i,j\}\in E(g)}\hat{f}_{i,j}\right|\prod_{i=2}^n z_{h}(S_i)\le e^{|[S_1]_1|},\nonumber
	\label{Corollary}
\end{equation}  
under the condition (Section 5.7.1 in \cite{friedli2017statistical})
\begin{equation}
	\sum_{S\subset\Lambda}|w_h(S)z_h^{|S|}||\hat{f}(S,S^*)|e^{|[S]_1|}\le|[S^*]_1|,\;\;\forall\;S^*\subset\Lambda.
	\label{Conv-cond}
\end{equation}
Hence, having that  $|[S]_1|\le (2d+1)|S|$ and
\begin{eqnarray}
	\sum_{S\subset\Lambda}|w_h(S)z_h^{|S|}||\hat{f}(S,S^*)|e^{|[S]_1|}\le |[S^*]_1|\sum_{S\ni0}|w_h(S)z_h^{|S|}|e^{|[S]_1|}\nonumber
	\\
	\le|[S^*]_1|\sum_{n\ge1}e^{n[2\beta h+2d+1+2\log(2d)]},\nonumber
\end{eqnarray}
\eqref{Conv-cond} is valid when
\begin{equation}
	K(h,d):=[e^{-(2\beta h +2d+1+2\log(2d))}-1]^{-1}=\sum_{n\ge1}e^{n[2\beta h+2d+1+2\log(2d)]}\le1,\nonumber
	\label{Conv_IS1}
\end{equation} 
i.e., for all $h$ such that $h\le h_{IS}:=-\frac{1}{2\beta}\left(2d+1+2\log(2d)+\log2\right).$
From \eqref{MU-H}, the corresponding chemical potential is given by the following: 
\begin{equation}
	\mathcal{M}_{IS}:=2h_{IS}-4dJ.
	\label{FVBoundCluster}
\end{equation}

For the cluster expansion of $\Xi^{\z}_{\Lambda,\beta}(\mu_{\Lambda})$ we will use Theorem 1 in \cite{procacci2017convergence}, recalled below and adapted to our context ($\Lambda\subset\mathbb{Z}^d$ instead of $\Lambda\in\mathbb{R}^d$).
\begin{theorem}\label{ThP}[Theorem 1 in \cite{procacci2017convergence}] Let $V$ be a stable and tempered pair potential with stability constant $B$. Then 
	\begin{equation*}
		\left|\frac{1}{|\Lambda|}\frac{1}{n!}\sum_{\mathbf{x}\in\Lambda^n}\sum_{g\in\mathcal{C}_n}\prod_{\{i,j\}\in E(g)}f_{i,j}\right|\le e^{\beta Bn}n^{n-2}\frac{[\hat{C}(\beta)]^{n-1}}{n!},
	\end{equation*}		
	where 
	\begin{equation}\label{HatCbeta}
		\hat{C}(\beta):=\sum_{x\in\mathbb{Z}^d}\left[1-e^{-\beta|V(x)|}	\right].
	\end{equation}	
	Therefore, the Mayer series \[z+\sum_{n\ge2}\left[\frac{1}{|\Lambda|}\frac{1}{n!}\sum_{\mathbf{x}\in\Lambda^n}\sum_{g\in\mathcal{C}_n}\prod_{\{i,j\}\in E(g)}f_{i,j}\right]z^n,\]
	converges absolutely, uniformly in  $\Lambda$, for any complex $z$ inside the disk
	\[|z|< [e^{\beta B+1}\hat{C}(\beta)]^{-1},\]
	i.e., the convergence radius $R$ of the Mayer series admits the following lower bound 
	\[R\ge R^*:= [e^{\beta B+1}\hat{C}(\beta)]^{-1}.\]
\end{theorem}

Hence, having $z=e^{\beta \mu_{\Lambda}}$ and $\hat{C}(\beta)$ given by $\bar{C}_{J,d}(\beta)$  defined in \eqref{C-Bar}, when 
\begin{equation}
	e^{\beta\mu_\Lambda+\beta B}\bar{C}_{J,d}(\beta)<e^{-1}\;\Leftrightarrow\; \mu_{\Lambda}\le-\frac{1}{\beta}\log\left(e^{\beta B+1}\bar{C}_{J,d}(\beta)\right)=:\mathcal{M}_{LG},
	\label{Bound-Mu}
\end{equation}
where $B$ given by \eqref{Stability}, \eqref{GcPF1} can be written as
\begin{equation}
	\Xi^{\z}_{\Lambda,\beta}(\mu_{\Lambda})=\exp\left\{\sum_{N\ge1}\frac{e^{\beta\mu_{\Lambda}N}}{N!}\sum_{g\in\mathcal{C}_N}\sum_{\mathbf{x}\in\Lambda^N}\prod_{\{i,j\}\in E(g)}f_{i,j}\right\},
	\label{GC-ClusterE}
\end{equation}	 
where the series in the exponent is absolutely convergent. 

Below we compare $\mathcal{M}_{LG}$	and $\mathcal{M}_{IS}$ for fixed different values of $J$ and $d$ and with $\beta\in[0,1]$. In Figure \ref{CE2}, we  compare the two lower bounds of the radius of convergence for $J=1,2$, $d=1$ and $\beta\in[0,1]$. We observe that there exists $\bar{\beta}\equiv\bar{\beta}(d,J)$ such that $\mathcal{M}_{IS}\le\mathcal{M}_{LG}$ for all $\beta\le \bar{\beta}$ and  $\mathcal{M}_{LG}<\mathcal{M}_{IS}$ when $\beta>\bar{\beta}$. A similar behavior can also be observed if we fix $J$ and we consider different values for the dimension, as it is shown in Figure \ref{CE1}, where we considered $J=1$, $d=1,2$ and $\beta\in[0,1]$.

\begin{figure}[H]
	\centering
	\includegraphics[width=1
	\textwidth]{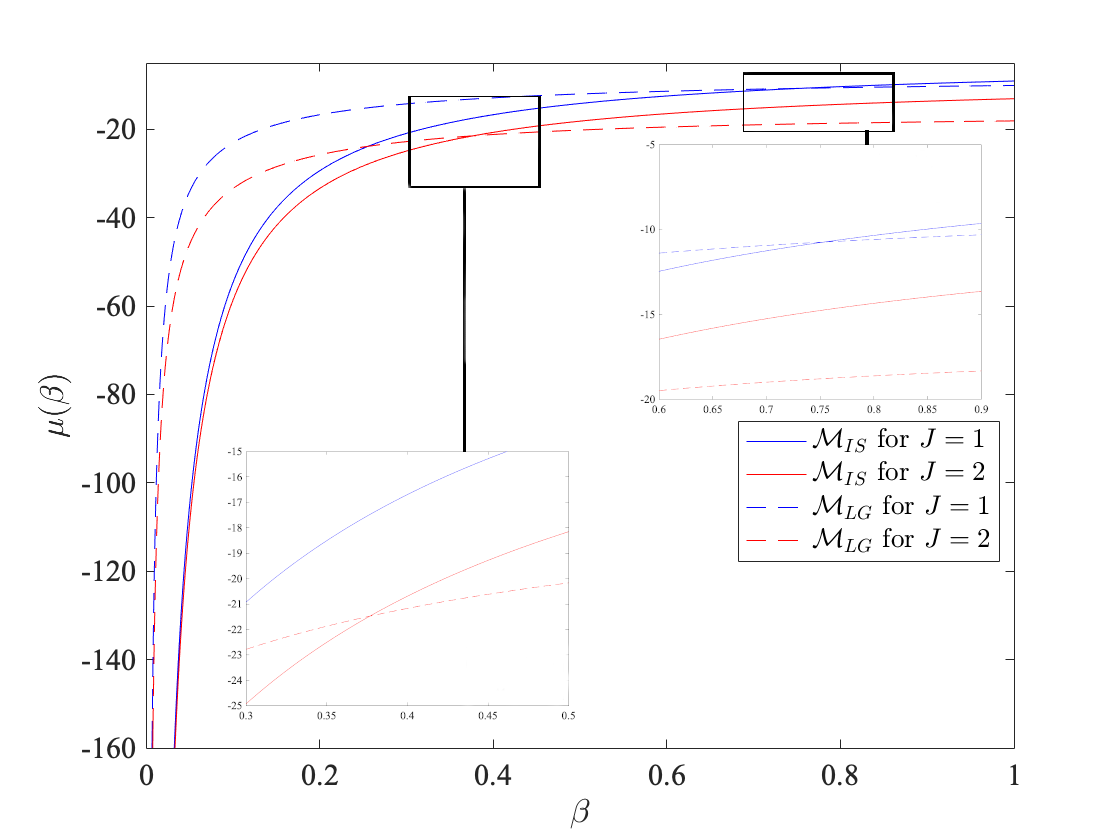}
	\caption{$\mathcal{M}_{IS}$ (continuous line) and $\mathcal{M}_{LG}$ (dashed line) in dimension 1, with $\beta\in[0,1]$ and $J=1$ (blue lines) and 2  (red lines).}
	\label{CE2}
\end{figure}	
\begin{figure}[H]
	\centering
	\includegraphics[width=1
	\textwidth]{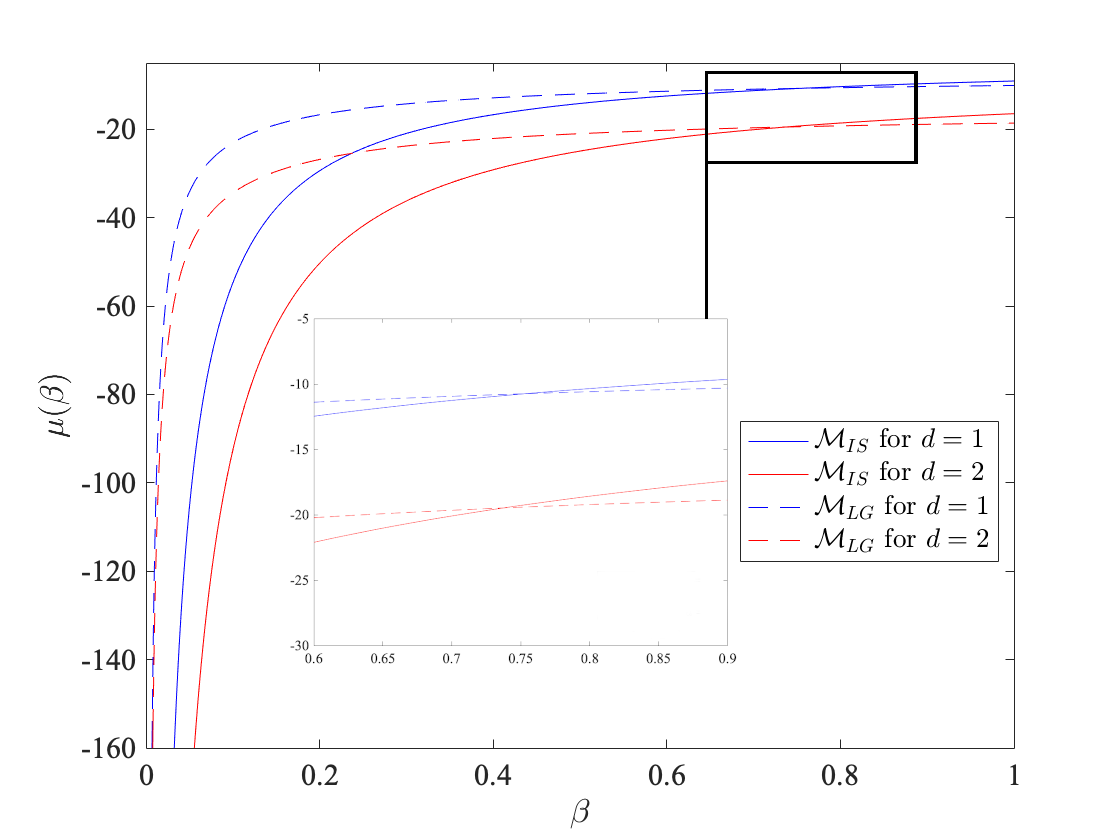}
	\caption{$\mathcal{M}_{IS}$ (continuous line) and $\mathcal{M}_{LG}$ (dashed line) with $J=1$, $\beta\in[0,1]$ and $d=1$ (blue lines) and 2  (red lines).}
	\label{CE1}
\end{figure}

\subsection{Virial inversion.}\label{SubSecGC2}

From \eqref{TdP} and \eqref{GC-ClusterE} we find
\begin{equation}
	\beta p_{\beta}(\mu)=\sum_{n\ge1}e^{\beta\mu n}b_n,
	\label{P-INF}
\end{equation}
with $\mu:=\lim_{\Lambda\to\infty}\mu_{\Lambda}=2h-4Jd$. The $b_n$'s are the \textquotedblleft discrete version\textquotedblright\space - in the same sense of \eqref{Mayer} - of the {\it{connected Mayer's coefficient}} (formula (13.5) in \cite{mayerstatistical}). More precisely they are defined as: \[b_n:=\frac{1}{n!}\sum_{g\in\mathcal{C}_n}\sum_{\mathbf{x}\in(\mathbb{Z}^d)^n}\prod_{\{i,j\}\in E(g)}f_{i,j}.\]

Hence, we derive now the density expansion for the pressure defined in \eqref{TdP} which can be written also as in  \eqref{P-INF}. We recall that this representation is equivalent with the one of the Ising model (see formulas \eqref{GcPF2}-\eqref{MU-H}). Moreover, thanks to this equivalence we have that between the thermodynamic pressure of the Ising model - \eqref{ThermodynPressIs} -  and the one of the lattice gas system - \eqref{TdP} - the relation below occurs: 
\begin{equation}
	\beta p_{\beta}(\mu)=\beta\psi_{\beta}(h)-\beta J d+\beta h.
	\label{Pmu-Ph}
\end{equation}

Let us define now the density as follows:
\begin{equation}
	\rho\equiv\rho(\mu):=\beta\frac{\partial p_{\beta}(\mu)}{\partial\log(e^{\beta\mu})}=\frac{\partial p_{\beta}(\mu)}{\partial\mu}.
	\label{GC-Density}
\end{equation}
%where from \eqref{P-INF} we have 
%\[\rho=\frac{\partial p_{\beta}(\mu)}{\partial\mu}=\sum_{n\ge1}ne^{\beta\mu n}b_n\]
%calling
Using the results presented in \cite{jansen2019virial}, we get: 
\begin{equation}
	\beta\mu\equiv\beta\mu(\rho)=\log\rho-\sum_{n\ge1}\beta_n\rho^n\;\;\mathrm{and}\;\;\beta p_{\beta}(\rho)=\rho+\sum_{n\ge1}\frac{n\beta_n}{n+1}\rho^{n+1},
	\label{InvMu}
\end{equation}
when 
\begin{equation}
	\rho\le \mathcal{R}_V:=\left(2e^{1+\beta[4J(2d+1)]}\bar{C}_{J,d}(\beta)\right)^{-1},
	\label{RadGCanL}
\end{equation}
where $\beta_n$'s are given by \eqref{Mayer} and $\bar{C}_{J,d}(\beta)$ is defined in \eqref{C-Bar}. 

Wanting to be more precise, the validity of \eqref{InvMu} under the condition \eqref{RadGCanL} follows from the application of Theorem 4.1 in \cite{jansen2019virial} recalled below, and as for Theorem \ref{ThP}, adapted to our context ($\Lambda\subset\mathbb{Z}^d$).We will call $B^*$ the positive constant such that $\inf V\ge -B^*$ - which is given in our case by  $B^*:=4J$ - and we will use the quantity $\hat{C}(\beta)$  defined in \eqref{HatCbeta} and $\beta_n$ given by \eqref{Mayer}.

\begin{theorem}[Theorem 4.1 in \cite{jansen2019virial}]
	
	(a) If $\rho\in \mathbb{C}$ satisfies $\hat{C}(\beta)e^{\beta[B+B^*]}|\rho|\le (2e)^{-1}$, then $\sum_{n\ge1}|\beta_n\rho^n|\le \frac{1}{2}$. In particular the radius of convergence of the previous sum is bounded by below by
	\[R^*_{V}:=\left[2e^{1+\beta[B+B^*]}\hat{C}(\beta)\right]^{-1}.\]
	
	(b) There exists a neighborhood $\mathcal{O}$ of the origin with 
	\[\left\{z\in\mathbb{C}\;|\; |z|e^{\beta[B+B^*]}\hat{C}(\beta)< \frac{1}{e e^{2/e}}\right\}\subset\mathcal{O}\subset\left\{z\in\mathbb{C}\;|\;|z|e^{\beta[B+B^*]}\hat{C}(\beta)< \frac{1}{2\sqrt{e}}\right\}\]
	such that $\rho\equiv\rho(z)$ is a bijection from $\mathcal{O}$ onto the open ball $B(0,R^*_V)$, with inverse 
	\[z(\rho)=\rho\exp\left\{-\sum_{n\ge1}\rho^n\beta_n\right\}.\]
	
	(c) For all $z\in\mathcal{O}$, we have
	\[\beta p_{\beta}(z)=\rho(z)+\sum_{n\ge1}\frac{n\beta_n}{n+1}[\rho(z)]^{n+1}.\]
	
	(d) For all $\rho\in B(0,R^*_V)$, the Helmhotz free energy $f_{\beta}(\rho):=\sup_{z}\{\beta^{-1}\rho\log z-p_{\beta}(z)\}$, is given by
	\[\beta f_{\beta}(\rho)=\rho(\log\rho-1)-\sum_{n\ge1}\frac{\rho^{n+1}}{n+1}\beta_n.\]
\end{theorem}

Following (d) of the previous theorem, i.e., from \eqref{TFE_CE} and \eqref{InvMu}, we can recover (explicitly) the Legendre transform relations between the thermodynamic free energy and the thermodynamic pressure given by \eqref{LeTr}, i.e. 
\[\rho(\log\rho-1)-\sum_{n\ge 1}\frac{\rho^{n+1}}{n+1}\beta_n=\rho\left[\log\rho-\sum_{n\ge1}\rho^n\beta_n\right]-\rho+\sum_{n\ge1}\frac{n\rho^{n+1}}{n+1}\beta_n.\]

In the next figures we compare $\mathcal{R}_C$ and $\mathcal{R}_V$ in dimension 1,2,3, with $J=1$ and $\beta\in[0,1]$ (Figure \ref{Virial2}). 
%as well in dimension 2 with $J=1,2,5$ and $\beta[0,1]$ (Figure \ref{Virial1}). 
We have that the lower bound obtained in the grand-canonical ensemble is bigger than the one obtained in the canonical ensemble.  Moreover, the same behavior can be observed if we fix $d$ and vary $J$. 
As it is deductible from Figures \ref{CanCan1}, %\ref{CanCan2}, 
we will find the same behavior if we consider $\bar{\mathcal{R}}_C$ instead of $\mathcal{R}_C$.

%As it will be discussed in detail in Subsection \ref{SectionRD}, one can also define the density as a function of the activity (equation \eqref{GC-Density}) and by inverting this formula, obtain an expression for the pressure with respect to the density, which is the so called {\it{virial expansion}}. In the figure below we compare the radius of  convergence this expansion, denoted with $\mathcal{R_V}^{LG}\equiv\mathcal{R_V}^{LG}(\beta,d,J)$,  with the one of the canonical expansion presented in Theorem \ref{TH1}. We can see that $\mathcal{R_V}^{LG}$ gives us a bigger convergence density region than $\mathcal{R}^C$. In Figure \ref{Virial2}, we represent $\mathcal{R_V}^{LG}$ and $\mathcal{R}^C$ with $J= 1$ and $\beta\in[0,1]$, in dimension 1, 2 and 3. We can also observe the same behavior if we fix the dimension and we consider different values of J as it is shown in Figure \ref{Virial1} in Subsection \ref{SectionRD}.  

\begin{figure}[H]
	\centering
	\includegraphics[width=1
	\textwidth]{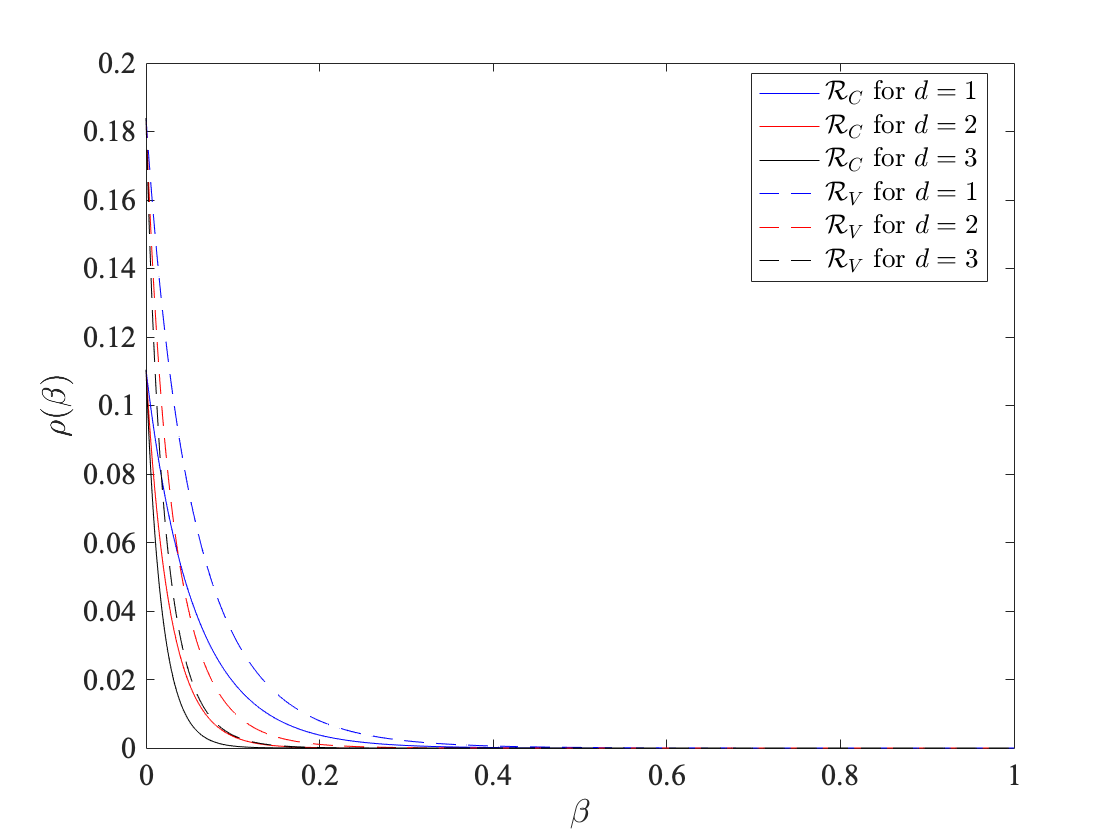}
	\caption{$\mathcal{R}_C$ (continuous line) and $\mathcal{R}_{V}$ (dashed line) with $J=1$ and $\beta\in[0,1]$ in dimension 1 (blue lines), 2 (red lines) and 3 (black lines).}
	\label{Virial2}
\end{figure}

%\begin{figure}[H]
%	\centering
%	\includegraphics[width=1
%	\textwidth]{Can-Gcan2.png}
%	\caption{$\mathcal{R}_C$ (continuous line) and $\mathcal{R}_V$ (dashed line) in dimension 1 with $J=1$ (blue lines), 2 (red lines),  5 (black lines) and $\beta\in[0,1]$.}
%	\label{Virial1}
%\end{figure}}

%\begin{figure}[H]
%	\centering
%	\includegraphics[width=0.9
%	\textwidth]{Can-Gcan3.png}
%	\caption{$\bar{\mathcal{R}}_C$ (continuous line) and $\mathcal{R}_V$ (dashed line) in dimension 1 with $J=1$ (blue lines), 2 (red lines),  5 (black lines) and $\beta\in[0,1]$.}
%	\label{Virial3}
%\end{figure}}

We want to conclude this section with the following observation.  
Let us now define the {\it{thermodynamic free energy for the Ising model}} as
\begin{equation}\label{freeEnIs}
	\phi_{\beta}(m):=\lim_{\Lambda\to\mathbb{Z}^d}-\frac{1}{\beta|\Lambda|}\log\tilde{Z}^{-}_{\Lambda,\beta}(m).
\end{equation}  
From \eqref{CpfLg} and \eqref{Cpf}  we obtain that:
\begin{equation}
	\beta f_{\beta}(\rho)=\beta\phi_{\beta}(m)-4d\beta J\left(\frac{m+1}{2}\right)+d\beta J,   
	\label{FreeE-Lattice-Ising}
\end{equation}
since $|\mathcal{E}_{\Lambda}|/|\Lambda|\to d$ as $|\Lambda|\to\infty$
and where $\rho=m'=(m+1)/2$. 

Furthermore, \eqref{LeTr}, \eqref{Pmu-Ph} and \eqref{FreeE-Lattice-Ising} give us the following relation between the thermodynamic free energy and the thermodynamic pressure for the Ising model
\begin{equation*}
	\beta\phi_{\beta}(m)=\sup_{h}\left\{2h\left(\frac{m+1}{2}\right)-\beta h-\beta\psi_{\beta}(h)\right\}.
\end{equation*}

\section{Decay of correlations in the canonical ensemble, proof of Theorem \ref{TH2}}
\label{SectionD}

For the proof of Theorem \ref{TH2} we follow the strategy of \cite{kuna2018convergence}.   Let $n\in\mathbb{N}_0$ and $k\in\mathbb{N}$. We denote with $\mathcal{C}_{n,n+k}$ the set of connected graphs with $n+k$ vertices, where  we singled out $n$  vertices which will be called \textquotedblleft white\textquotedblright\space and the remaining $k$ vertices will be called \textquotedblleft black\textquotedblright. %We denote with $\mathcal{C}_{n,n+k}$ the set of these graphs.
Moreover, we call \textit{articulation\;vertex}, a vertex such that removing it the graph is decomposed in two or more separate part, where at least one of them does not contain white vertices. Hence, we denote with  $\mathcal{B}^{AF}_{n,n+k}$  the set of graphs with $n$ white and $k$ black vertices and without articulation vertices.   

We define the \textit{n-point\;correlation\;function} with $n\le N$  as:
\begin{equation*}
	\rho^{(n)}_{\Lambda,N}(q_1,...,q_n):=\frac{1}{(N-n)!}\sum_{\mathbf{x}\in\Lambda^{N-n}}\frac{1}{Z^{per}_{\Lambda,\beta}(N)}e^{-\beta H^{per}_{\Lambda}(q_1,...,q_n,\mathbf{x})},
\end{equation*}
where with $\{q_i\}_{i=1}^n\subset\Lambda$ we denote the fixed particles and $Z^{per}_{\Lambda,\beta}(N)$ is given by \eqref{CanPF} with periodic boundary conditions. When we will do the cluster expansion, the fixed particles  $\{q_i\}_{i=1}^n$ will correspond to the white vertices in the connected graphs (clusters).  

Denoting with $\mu_{\Lambda,\beta ,N}(\cdot)$ the {\it{canonical Gibbs measure}} in the volume $\Lambda$, i.e.,
\begin{equation}
	\mu_{\Lambda,\beta,N}(C):=\frac{1}{Z^{per}_{\beta,\Lambda}(N)}\frac{1}{N!}\sum_{\mathbf{x}\in\Lambda^N\cap\; C}e^{-\beta H^{per}_{\Lambda}(\mathbf{x})},\nonumber
\end{equation}
where $C\subset (\mathbb{Z}^d)^{N}$, we define for a test function $\varphi$, the  {\it{Bogoliubov functional}} $L_B(\varphi)$ as
\begin{equation}
	L_B(\varphi):=\sum_{\mathbf{x}\in\Lambda^N}\prod_{k=1}^N(1+\varphi(x_k))\mu_{\Lambda,\beta, N}(\{\mathbf{x}\}).\nonumber
	\label{BOG}
\end{equation}
We can define implicitly the \textit{truncated\;n-point\;correlation\;function} $u^{(n)}_{\Lambda,N}(\cdot)$ by its generating function which is the logarithm of the Bogoliubov functional, i.e.,
\begin{equation}
	\log L_{B}(\varphi)=:\sum_{n\ge1}\frac{1}{n!}\sum_{\mathbf{x}\in\Lambda^n}\varphi(x_1)\cdot\cdot\cdot\varphi(x_n)u^{(n)}_{\Lambda,N}(x_1,...,x_n),
	\label{trunc}
\end{equation} 
where, for example, when $n=2$ and fixing $q_1,q_2\in\Lambda,\;u^{(2)}_{\Lambda,N}(q_1,q_2)$ is given by \eqref{U2}.

The {\it{extended (canonical) partition function}} is defined as
\begin{equation}\begin{split}
		Z^{per}_{\Lambda,\beta,N}(\alpha\varphi):=\frac{1}{N!}\sum_{\mathbf{x}\in\Lambda^N}\prod_{i=1}^{N}(1+\alpha\varphi(x_i))e^{-\beta H^{per}_{\Lambda}(\mathbf{x})}\nonumber
	\end{split}
	\label{ModCan}
\end{equation}
with $\alpha\in\mathbb{R}$, such that 
\begin{equation}
	L_B(\alpha\varphi)=\frac{Z^{per}_{\Lambda,\beta,N}(\alpha\varphi)}{Z^{per}_{\Lambda,\beta,N}(0)},\;\;\mathrm{where}\;\;Z^{per}_{\Lambda,\beta,N}(0)\equiv Z^{per}_{\Lambda,\beta}(N)\nonumber
\end{equation}
and then, thanks to \eqref{trunc} for all $n\ge1$, we have
\begin{equation}
	\sum_{\mathbf{x}\in\Lambda^n}\varphi(x_1)\cdot\cdot\cdot\varphi(x_n)u^{(n)}_{\Lambda,N}(x_1,...x_n)=\frac{\partial^n}{\partial \alpha^n}\log Z^{per}_{\Lambda,\beta,N}(\alpha\varphi)\bigg|_{\alpha=0}.
	\label{Key}
\end{equation}

Using the polymer model representation recalled in Section \ref{SectionCE}, with set of polymers 
$\mathcal{V}^*_{N}:=\{\{(V_1,A_1),...,(V_k,A_k)\}\;|\;V_i\in\{1,...,N\},\;|V_i|\ge2,\;{\rm{and}}\;A_i\subset V_i\;\forall\;i=1,...,k\}$ where the compatibility relation is here given by $(V_i,A_i)\sim(V_j,A_j)\Leftrightarrow V_i\cap V_j=\emptyset$ and with weights
\begin{equation}
	\bar{\zeta}_{\Lambda}((V,A)):=\alpha^{|A|}\sum_{g\in\mathcal{C}_{V}}\frac{1}{|\Lambda|^{|V(g)|}}\sum_{\mathbf{x}\in\Lambda^{|V(g)|}}\prod_{\{i,j\}\in E(g)}f_{i,j}\prod_{i\in A}\varphi(x_i),\nonumber
\end{equation}	
for $N/|\Lambda|$ small enough (see Theorem 2.1 in \cite{kuna2018convergence}), we have
\begin{eqnarray}
	\log Z^{per}_{\Lambda,\beta,N}(\alpha\varphi)&=&\log Z^{per}_{\Lambda,\beta}(N)\nonumber\\
	&+&\sum_{n=1}^N\sum_{m=1}^n\sum_{k=0}^{N-m}{N\choose{m+k}}{{m+k}\choose m}\alpha^n\sum_{\substack{I\;:\;\bigcup_{(V,A)\in\mathrm{supp}I}A=[m]\\\bigcup_{(V,A)\in\mathrm{supp}I}V=[m+k]\\\sum_{(V,A)\in\mathrm{supp}I} |A| I((V,A))=n}} c_I\bar{\zeta}_{\Lambda}^I \nonumber
	\\
	&=& \log Z^{per}_{\Lambda,\beta}(N)+\sum_{n=1}^N\sum_{m=1}^n\sum_{k=0}^{N-m}\alpha^n \tilde{P}_{N,|\Lambda|}(m+k)\tilde{B}_{\Lambda,\beta}(n,m,k),
	\label{ClusterModCanInt}
\end{eqnarray}  
where 
\begin{equation}\label{P}
	\tilde{P}_{N,\Lambda}(n):=\begin{cases}
		\frac{N(N-1)\cdot\cdot\cdot(N-n+1)}{|\Lambda|^n},\;\;\;\;\;\mathrm{for}\;n\le N,\\
		0,\;\;\;\;\;\;\;\;\;\;\;\;\;\;\;\;\;\;\;\;\;\;\;\;\;\;\;\;\;\;\mathrm{otherwise},\nonumber
	\end{cases}
\end{equation}
and
\begin{equation}
	\tilde{B}_{\Lambda,\beta}(n,m,k):=\frac{|\Lambda|^{(m+k)}}{m!k!}\sum_{\substack{I\;:\;\bigcup_{(V,A)\in\mathrm{supp}I}A=[m]\\\bigcup_{(V,A)\in\mathrm{supp}I}V=[m+k]\\\sum_{(V,A)\in\mathrm{supp}I} |A| I((V,A))=n}} c_I\bar{\zeta}_{\Lambda}^I.\nonumber
	\label{ModMayer}
\end{equation}

The term $\tilde{B}_{\Lambda,\beta}(n,m,k)$ can be written as
\begin{equation}
	\tilde{B}_{\Lambda,\beta}(n,m,k)=\bar{B}_{\Lambda,\beta}(n,k)\delta_{n,m}+R_{\Lambda,\beta}(n,m,k)
	\label{Correlation4}
\end{equation}
with
\begin{equation}
	\bar{B}_{\Lambda,\beta}(n,k):=\frac{|\Lambda|^{(n+k)}}{n!k!}\sum_{I\;:\:A(I)=[n+k]}^*c_I\bar{\zeta}_{\Lambda}^I=\frac{1}{n!k!}\sum_{g\in\mathcal{B}^{AF}_{n,n+k}}\sum_{\mathbf{x}\in\Lambda^{n+k}}\prod_{\{i,j\}\in E(g)}f_{i,j}\prod_{i=1}^n\varphi(x_i).
	\label{Correlation5}
\end{equation}
In the previous definition with $*$ we mean that the sum runs over all multi-indices which satisfy  $n+k=|V_0|+\sum_{(V,A)\in\mathrm{supp}I,\;V\ne V_0}(|V|-1)$ and $I((V,A))=1$ for all $(V,A)\in\mathrm{supp}I$ and where $V_0$ contains the indices $1,2,...,n$. The second form of $\bar{B}_{\Lambda,\beta}(n,k)$ expressed in \eqref{Correlation5} is due to the fact that we consider here periodic boundary conditions (Lemma 4.1 in \cite{kuna2018convergence}).

Hence, from \eqref{Key}, \eqref{ClusterModCanInt} and \eqref{Correlation4} we get:
\begin{eqnarray}
	\frac{1}{2}\sum_{(x_1,x_2)\in\Lambda^2}\varphi(x_1)\varphi(x_2)u^{(2)}_{\Lambda,N}(x_1,x_2)&=&\sum_{k=0}^{N-1}\tilde{P}_{N,|\Lambda|}(1+k)R_{\Lambda,\beta}(2,1,k)\nonumber
	\\
	&+&\sum_{k=0}^{N-2}\tilde{P}_{N,|\Lambda|}(2+k)\bar{B}_{\Lambda,\beta}(2,k).\;\;\;
	\label{Correlation2}
\end{eqnarray}

The first sum gives a contribution of order $|\Lambda|^{-1}$. This estimate comes from the fact that the term $R_{\Lambda,\beta}(n,m,k)$ consists of lower order terms and in particular from \cite{pulvirenti2015finite} (as it is also recalled \cite{kuna2018convergence}), we have
\begin{equation}
	|R_{\Lambda,\beta}(n,m,k)|\le C\frac{1}{|\Lambda|},   
	\label{BoundErr}
\end{equation}
for all $n,\;k$ and uniformly on $\varphi$. 

For the first term ($k=0$) in the second sum ($n=m=2$) we have
\begin{eqnarray}
	&&\frac{1}{2}\left|\sum_{(x_1,x_2)\in\Lambda^2}\frac{N(N-1)}{|\Lambda|^2}f_{1,2}\;\varphi(x_1)\varphi(x_2)\right|\nonumber
	\\
	&&\le\frac{1}{2}\left[\left(\frac{N}{|\Lambda|}\right)^2+\frac{N}{|\Lambda|^2}\right]\sum_{(x_1,x_2)\in\Lambda^2}|\varphi(x_1)\varphi(x_2)|\left[(e^{4\beta J}-1)\mathbf{1}_{\{|x_1-x_2|=1\}}+\mathbf{1}_{\{x_1=x_2\}}\right]\;\;\;\;\;\;\;\;	
	\label{Correlation8}
\end{eqnarray}

For $k\ge1$ we will use the analogous of Lemma 4.2 (which is recalled below) in \cite{kuna2018convergence} in order to exchange the sum over $k$ and the one over $\mathbf{x}$.
\begin{lemma}\label{Lemma4.1}
	For any $n\ge2$ and  k$\ge1$ we have that 
	\begin{equation}
		\tilde{P}_{N,|\Lambda|}(n+k)\hat{B}_{\Lambda,\beta}(n;k)\le C \left(\frac{N}{|\Lambda|}\right)^2 e^{-ck},\nonumber
	\end{equation}
	where
	\begin{equation}
		\hat{B}_{\Lambda,\beta}(n;k):=\frac{1}{n!k!}\sum_{\mathbf{x}\in\Lambda^k}\left|\sum_{g\in\mathcal{B}^{AF}_{n,n+k}}\prod_{\{i,j\}\in E(g)}f_{i,j}\right|,\nonumber
	\end{equation}
	for some $c>1$ and $C>0$ independent on $k,\;N$ and $\Lambda$.
\end{lemma}
\begin{proof}
	The proof follows immediately from \cite{kuna2018convergence}. Indeed the calculation is similar to the one presented by the authors for the proof of Lemma 4.2 and the fact that we can choose $c>1$ is possible thanks to their Theorem 3.1.
\end{proof}

Moreover we multiply and divide for $e^{|x_1-x_2|}$. Hence from the fact that $|x_1-x_2|\le |V_0|-1\le k$ and using the second equality of \eqref{Correlation5}, we find
\begin{eqnarray}
	&&\left|\sum_{k=1}^{N-2}\tilde{P}_{N,|\Lambda|}(2+k)\bar{B}_{\Lambda,\beta}(2,k)\right|\nonumber
	\\
	&&\;\;\;\;\;\le\frac{1}{2}\sum_{(x_1,x_2)\in\Lambda^2}|\varphi(x_1)\varphi(x_2)| e^{-|x_1-x_2|}e^{|V_0|}\sum_{k=1}^{N-2}\tilde{P}_{N,|\Lambda|}(2+k)\hat{B}_{\Lambda,\beta}(2;k)\nonumber
	\\
	&&\;\;\;\;\;\le\frac{C}{2}\left(\frac{N}{|\Lambda|}\right)^2\sum_{(x_1,x_2)\in\Lambda^2}|\varphi(x_1)\varphi(x_2)| e^{-|x_1-x_2|}\sum_{k=1}^{N-2}e^{-(c-1)k}\nonumber
	\\
	&&\;\;\;\;\;\le\frac{C_1}{2}\left(\frac{N}{|\Lambda|}\right)^2\left[\sum_{(x_1,x_2)\in\Lambda^2}|\varphi(x_1)\varphi(x_2)|e^{-|x_1-x_2|}\right],
	\label{lastE}
\end{eqnarray}
where $c$ and $C$ are the constants of Lemma \ref{Lemma4.1} and $C_1$ is a positive constant bigger than $C$ and independents on $N,\Lambda$.  

Then from \eqref{Correlation2}, \eqref{BoundErr}, \eqref{Correlation8} and \eqref{lastE} we have
\begin{eqnarray}
	&&\sum_{(x_1,x_2)\in\Lambda^2}|\varphi(x_1)\varphi(x_2)||u^{(2)}_{\Lambda,N}(x_1,x_2)|\nonumber
	\\
	&&\;\;\;\;\;\le\sum_{(x_1,x_2)\in\Lambda^2}|\varphi(x_1)\varphi(x_2)|
	\bigg\{\left(\frac{N}{|\Lambda|}\right)^2\bigg[(e^{4\beta J}-1)\mathbf{1}_{\{|x_1-x_2|=1\}}+\mathbf{1}_{\{x_1=x_2\}}\nonumber
	\\
	&&\;\;\;\;\;\;\;\;\;\;\;\;\;\;+\frac{(e^{4\beta J}-1)\mathbf{1}_{\{|x_1-x_2|=1\}}+\mathbf{1}_{\{x_1=x_2\}}}{N}+C e^{-|x_1-x_2|}\bigg]+ C_1\frac{1}{|\Lambda|}\bigg\}\;\;\;\;\;\;\;\;\;\;
\end{eqnarray}
with $C,C_1\in\mathbb{R}^+$. Then the conclusion follows choosing as test functions the Kronecker deltas in $q_1$ and $q_2$.

\section{Precise large and local moderate deviations, proofs of Theorems \ref{TH-LD}, \ref{TH-MD} and Corollary \ref{CLT}}
\label{SectionLMD}

In this section we compare our approach for the study of precise large  and local moderate deviations (Theorems \ref{TH-LD}, \ref{TH-MD} and Corollary \ref{CLT}) with  the ones presented in \cite{del1974local} and, in particular, in \cite{dobrushin1994large}. The proofs of the Theorems are the ones given in \cite{scola2020local} (recalled in Appendix \ref{S2}), since once one can write
$\log Z^{\g}_{\Lambda,\beta}$ as a power series of the density (Theorem \ref{TH1})
then the proof is the same. Note that, thanks to \eqref{HamIs}, \eqref{Ham2}, \eqref{NewH}, \eqref{GcPF2} - \eqref{Mu-H}, the probability defined in \eqref{Prob1} can be expressed via the grand-canonical probability measure for the Ising model with $-1$ boundary conditions.

\begin{proof}[Proof of Theorem \ref{TH-LD}]
	The proof follows from Theorem \ref{TH1} and the proof of Theorem 2.1 in \cite{scola2020local} and is recalled in Appendix \ref{S2}. 
\end{proof}

\begin{proof}[Proof of Theorem \ref{TH-MD}]
	The proof follows from Theorem \ref{TH1} and  the proof of Theorem 2.2 in \cite{scola2020local} and is recalled in Appendix \ref{S2}. 
\end{proof}	
\begin{proof}[Proof of Corollary \ref{CLT}]
	The proof follows from Theorem \ref{TH-MD} for $\alpha=1/2$.
\end{proof}

In order to do the comparison,  we briefly recall  the approach 
followed in \cite{del1974local}  and \cite{dobrushin1994large}. For a fixed chemical potential $\mu_0$, we define the logarithmic generating function for the moments at finite volume associated to the probability given by  \eqref{Prob1} as
\begin{equation}
	L^{\z}_{\Lambda,\beta, \mu_0}(\mu):=\log\left[\sum_{N\ge0}\probBC(A_N)e^{\beta\mu N}\right],
	\label{H}
\end{equation}
with $A_N$ given by \eqref{DeviationSet}. From \eqref{MeanV} and \eqref{H} we have 
\begin{equation}
	\bRL=\frac{1}{|\Lambda|}\frac{1}{\beta}\frac{d}{d\mu}L^{\z}_{\Lambda,\beta,\mu_0}(\mu)\bigg|_{\mu=0}
\end{equation}
and
\begin{equation}
	\sigma^2_{\Lambda,\z}(\mu_0):=\mathbb{E}^{\z}_{\Lambda,\mu_0}\left[\frac{(N-\bRL|\Lambda|)^2}{|\Lambda|}\right]=
	\frac{1}{|\Lambda|}\frac{1}{\beta^2}\frac{d^2}{d\mu^2}L^{\z}_{\Lambda,\beta,\mu_0}(\mu)\bigg|_{\mu=0}.
	\label{2-mom}
\end{equation}
In general, denoting  by $G^{m}_{\Lambda,\z}$ the $m$-th moment per unit of volume, we have:
\begin{equation}
	G^{m}_{\Lambda,\z}=\frac{1}{\beta^{m}}\frac{d^m}{d\mu^m}L^{\z}_{\Lambda,\beta,\mu_0}(\mu)\bigg|_{\mu=0}.
	\label{m-mom}
\end{equation}
Let us define the characteristic function as
\begin{equation}\label{char}
	\varphi_{\Lambda,\mu'}(t):=\sum_{N\ge0}\mathbb{P}_{\Lambda,\mu'}^{\z}(A_{N})e^{itN},
\end{equation} 
where for $\mu'=\mu+\mu_0$, the ``excess (by $\mu$) probability measure" is given by
\begin{equation}\label{char2}
	\mathbb{P}_{\Lambda,\mu+\mu_0}^{\z}(A_{N}):=\exp\left\{-L^{\z}_{\Lambda,\beta,\mu_0}(\mu)+\beta\mu N\right\}\probBC(A_{N}).\nonumber
\end{equation}

First, for the large deviations, i.e., considering a deviation $\tN$ given by \eqref{Deviation} with $\alpha=1$, the probability of $A_{\tN}$ can be expressed using the excess measure optimizing over $\mu$
such that $\mathbb{P}_{\Lambda,\mu+\mu_0}^{\z}(A_{\tilde N})\sim 1$, i.e., by making $\tN$ ``central" with respect to the new measure. In this way we obtain 
\begin{equation}\begin{split}
		\lim_{|\Lambda|\to\infty}(\beta|\Lambda|)^{-1}\log\mathbb{P}^{\z}_{\Lambda,\mu_0}\left(A_{\tN}\right)=\lim_{|\Lambda|\to\infty}\frac{-\mathcal{I}^{\z}_{\Lambda,\beta,\mu_0}(\tN)}{\beta|\Lambda|}=- \frac{I_\beta(\tilde\rho;\rho_0)}{\beta},\nonumber
	\end{split}
	\label{Classical-LDP}
\end{equation}  
where 
\begin{equation}
	\mathcal{I}^{\z}_{\Lambda,\beta,\mu_0}(\tN):=\sup_{\mu}\left\{\beta\mu \tN-L^{\z}_{\Lambda,\beta,\mu_0}(\mu)  \right\}
	\label{I}
\end{equation}
and
\begin{equation}
	I_\beta(\tilde\rho;\rho_0):=\beta f_{\beta}(\tilde\rho)-\beta f_{\beta}(\rho_0)-\beta f'_{\beta}(\rho_0)(\tilde\rho-\rho_0).\nonumber
	\label{LD-F}
\end{equation}
In the previous  formulas  we have that the quantity $\rho_0$, which is the limit of $\bRL$ as $\Lambda\to\infty$, is also such that 
\begin{equation}
	f'_{\beta}(\rho_0)=\mu_0\Leftrightarrow p'_{\beta}(\mu_0)=\rho_0,\nonumber
	\label{Mu0Rho0}
\end{equation}
where the last relations follow from \eqref{LeTr}, \eqref{LeTr1} and the fact that we are far from the phase transition (see also \eqref{MeanV}). Moreover, let us note that from \eqref{bN-N*} $\rho_0$ is also the thermodynamic limit of $\rho^*_{\Lambda}$.  For later use, from \eqref{GcFreeE}, we have that \eqref{LDO} is the \textquotedblleft volume normalized version\textquotedblright of \eqref{I}, i.e., 
\begin{equation}
	I^{GC}_{\Lambda,\beta,\z}(\tRL;\bRL)=\mathcal{I}^{\z}_{\Lambda,\beta,\mu_0}(\tN)|\Lambda|^{-1}.
	\label{FV-I}
\end{equation}

A more precise formula at finite volume as well as the higher order corrections terms come from the inversion of \eqref{char}:
\begin{equation}\label{invert}
	\mathbb{P}_{\Lambda,\tilde\mu_{\Lambda}}^{\z}(A_{\tilde N})=
	\frac{1}{2\pi}\int_{-\pi}^{\pi}e^{-it\tN}\varphi_{\Lambda,\tilde\mu_{\Lambda}}(t)dt
\end{equation} 
where by $\tilde\mu_{\Lambda}$ we denote the optimal chemical potential found in \eqref{I}. This is also the approach of \cite{del1974local} for $\alpha=1/2$ and $\tilde{\mu}_{\Lambda}=\mu_0$. In \cite{del1974local} and \cite{dobrushin1994large} the authors, starting from \eqref{invert}, recover the inversion of the characteristic function of the Gaussian distribution which gives them, calculating the integral, a finite volume formula with an approximation for the high order correction term. This can be done  by the Taylor expansion at the second order of the characteristic function around $t=0$ and applying, for instance, the Gnedenko's method to estimate the integral. In particular, we refer to equations (4.1)-(4.10) of Section 4 in \cite{del1974local} and equations (2.1.30)-(2.1.34) Subsection 2.1 in \cite{dobrushin1994large}. On the other hand, our results come from a direct approach without passing from the calculation of the integral in \eqref{invert}, which also gives us an explicit formulation of the error terms. In fact, considering Theorem \ref{TH-LD}, the numerator in the fraction in the left hand side of \eqref{P3} comes immediately from definition \eqref{GcFreeE} and the Radon-Nikodyn derivative of our probability measure with respect to the one with $\tilde{\mu}_{\Lambda}$ (instead of $\mu_0$). This can be clearly observed in equation (2.53)  in \cite{scola2020local}.  Moreover, thanks  to the explicit formula that we have for the finite volume free energy  (Theorem \ref{TH1}), together
with $Z^{\z}_{\Lambda,\beta}(\tN)=\exp\left\{-\beta|\Lambda|f_{\Lambda,\beta,\z}(\tN)\right\}$ and \eqref{ProbA}, 
we can also obtain in an explicit and direct way both the normalization as well as the error terms  
as it is shown in Lemmas 6.3 in \cite{scola2020local}.

Second, starting from \eqref{char2} and considering the approach expressed in \cite{dobrushin1994large}, one can go a step further and study the local moderate deviations ($\alpha\in[1/2, 1)$ in \eqref{Deviation} by taking the Taylor expansion of \eqref{I} around $\bRL|\Lambda|$ and obtaining:
\begin{equation}
	\mathcal{I}^{\z}_{\Lambda,\beta,\mu_0}(\tN)=\frac{\beta^2(\tN-\bRL|\Lambda|)^2}{2|\Lambda|\sigma^2_{\Lambda,\z}(\mu_0)}+\sum_{j\ge3}\frac{Q^{(j)}_{\Lambda,\z}}{j!}\left(\frac{\tN-\bRL|\Lambda|}{|\Lambda|}\right)^j,
	\label{IT}
\end{equation}
where the coefficients $Q^{(j)}_{\Lambda,\z}$ are polynomials which can be computed via the moments \eqref{m-mom} as it is explained next.  In the previous equation we used \eqref{FV-I} and the fact that, from \eqref{GcFreeE}, we have $(f^{GC}_{\Lambda,\beta,\z})''(\bRL)=[p''_{\Lambda,\beta,\z}(\mu_0)]^{-1}=\beta[\sigma^2_{\Lambda,\z}(\mu_0)]^{-1}$. Note also that in \eqref{IT} we do not have the terms $\mathcal{I}^{\z}_{\Lambda,\beta,\mu_0}(\bRL|\Lambda|)$ and $(\mathcal{I}^{\z}_{\Lambda,\beta,\mu_0})'(\bRL|\Lambda|)$. This happens because, using the fact that $L^{\z}_{\Lambda,\beta,\mu_0}(\mu)$ is a strictly convex function of $\mu$, the supremum in \eqref{I} is obtained at $\mu=0$ when we consider $\bRL|\Lambda|$ instead of $\tN$ (see also \eqref{FV-I}).

In \cite{dobrushin1994large} (equations (1.2.18)-(1.2.23)), the polynomials $Q^{(j)}_{\Lambda,\z}$ are calculated substituting
\begin{equation}
	\beta(\tN-\bRL|\Lambda|)= (L^{\z}_{\Lambda,\beta,\mu_0})'(\tilde{\mu}_{\Lambda})-(L^{\z}_{\Lambda,\beta,\mu_0})'(0)=\beta\tilde{\mu}_{\Lambda} \sigma^2_{\Lambda,\z}(\mu_0)|\Lambda|+\sum_{m\ge3}\frac{(\beta\tilde{\mu}_{\Lambda})^{m-1}G^{m}_{\Lambda,\z}}{(m-1)!}
	\label{DS1}
\end{equation}
in 
\begin{equation}\label{DS2}
	\tilde{\mu}_{\Lambda}=(\mathcal{I}^{\z}_{\Lambda,\beta,\mu_0})'(\tN)=\frac{\tN-\bN}{\sigma^{2}_{\Lambda,\z}(\mu_0)|\Lambda|}+\sum_{m\ge3}Q^{(m)}_{\Lambda,\z}\frac{(\tN-\bN)^{m-1}}{(m-1)!},
\end{equation}
where  $(\mathcal{I}^{\z}_{\Lambda,\beta,\mu_0})'(\tN)$ is given by in \eqref{firstDeriv} for $x=\tN$, so that one obtains 
\begin{equation}
	Q^{(m)}_{\Lambda,\z}\equiv P\left(\frac{G^{3}_{\Lambda,\z}}{\sigma^{2}_{\Lambda,\z}(\mu_0)|\Lambda|},...,\frac{G^{m}_{\Lambda,\z}}{\sigma^{2}_{\Lambda,\z}(\mu_0)|\Lambda|},Q^{(3)}_{\Lambda,\z},...,Q^{(m-1)}_{\Lambda,\z}\right),\nonumber
\end{equation}  
where $P(x_1,..,x_n)$  is a polynomial in $x_1,...,x_n$. For example
\begin{equation}
	Q^{(3)}_{\Lambda,\z}=\frac{-G^3_{\Lambda,\z}}{(\sigma^{2}_{\Lambda,\z}(\mu_0)|\Lambda|)^3}\;\;\;{\rm{and}}\;\;Q^{(4)}_{\Lambda,\z}=\frac{-G^4_{\Lambda,\z}}{(\sigma^{2}_{\Lambda,\z}(\mu_0)|\Lambda|)^4}+3\frac{(G^{3}_{\Lambda,\z})^2}{(\sigma^{2}_{\Lambda,\z}(\mu_0)|\Lambda|)^5}.\nonumber
	\label{Firsts-Terms}
\end{equation}
We observe that also in this case our results follow directly from Theorem \ref{TH1}. Indeed - using \eqref{CeCan} - from \eqref{FVfreeE}, \eqref{FeRho}, \eqref{GcPF1} and \eqref{ProbA} the proofs of Theorem \ref{TH-LD} and especially Theorem \ref{TH-MD} and Corollary \ref{CLT} follow from the  Taylor expansion of the free energy defined in \eqref{ModFreeE} around $\rho^*_{\Lambda}$ (instead of the above indirect procedure,  i.e. via \eqref{DS1} and \eqref{DS2}). This can be seen from the fact that the main quantities involved - $D^{\alpha}_{\Lambda, \z}(\rho^*_{\Lambda})$ and $E_{|\Lambda|}(\alpha,u',\rho^*_{\Lambda})$ given by \eqref{Var2} and \eqref{Error} - are defined in terms of derivatives of $\F(\cdot)$.
Furthermore,   these derivatives are a version of the $Q^{(m)}_{\Lambda,\z}$'s in the canonical ensemble that are also equivalent in the thermodynamic limit.

%Indeed, this can be noted immediately by the fact that  $f_{\Lambda,\beta,\z}(\tN)$ and $\F(\tRL)$ as well as $\mu_0$ and $\F'(\rho^*_{\Lambda})$ are related via \eqref{FeRho} and 
%\eqref{Mu0Rho*} where, moreover, our probability can be written as in \eqref{ProbA}, i.e. \[\mathbb{P}_{\Lambda,\mu_0}^{\z}(A_{\tN})=[e^{\beta\mu_0\tN}\canBC(\tN)][\GcanBC(\mu_0)]^{-1},\] 
%. Hence  comes directly from the expansion of $\F(\tRL)$, around $\rho^*_{\Lambda}$ and the \textquotedblleft discrete Gaussian integral\textquotedblright\space which can be recovered in the same way from $\Xi_{\Lambda,\beta}^{\z}(\mu_0)$. For more details about this we refer in particular to Lemmas 3.1, 3.2 in \cite{scola2020local}. 

Moreover, another way to recover the $Q^{(m)}_{\Lambda,\z}$'s without using \eqref{DS1} and \eqref{DS2}, is given by the following remark.

\begin{remark}
	Let us note that another way for determining the terms $Q_{\Lambda,\z}^{(j)}$, can be derived directly from \eqref{I}. Indeed, let us define for all $x\in\mathbb{R}^+$ the function 
	\begin{equation}
		x\mapsto \mathcal{J}^{\z}_{\Lambda,\beta,\mu_0}(x):=\sup_{\mu\in\mathbb{R}}\left\{\beta x\mu-L^{\z}_{\Lambda,\beta,\mu_0}(\mu)\right\}=\beta x\mu(x)-L^{\z}_{\Lambda,\beta,\mu_0}(\mu(x)),
	\end{equation}
	where $\mu(x)$ is implicitly defined by  $\beta x=(L^{\z}_{\Lambda,\beta,\mu_0})'(\mu)$. Note that, when $x=N\in\mathbb{N}$, we get $\mathcal{J}^{\z}_{\Lambda,\beta,\mu_0}(x)=\mathcal{I}^{\z}_{\Lambda,\beta,\mu_0}(N)$, which happens if and only if $N=(L^{\z}_{\Lambda,\beta,\mu_0})'(\mu(N))$. 
	Hence we have:
	\begin{equation}
		(\mathcal{J}^{\z}_{\Lambda,\beta,\mu_0})'(x)=\beta\mu(x)=\beta\mu(x)+\mu'(x)[\beta x-(L^{\z}_{\Lambda,\beta,\mu_0})'(\mu(x))]
		\label{firstDeriv}
	\end{equation}
	and 
	\begin{eqnarray}
		(\mathcal{J}^{\z}_{\Lambda,\beta,\mu_0})''(x)&=&\beta \mu'(x)=2\beta \mu'(x)-(\mu'(x))^2(L^{\z}_{\Lambda,\beta,\mu_0})''(\mu(x))\nonumber
		\\
		&+&\mu''(x)[\beta x-(L^{\z}_{\Lambda,\beta,\mu_0})'(\mu(x))],\nonumber
	\end{eqnarray}
	which gives
	\begin{equation}
		(\mathcal{J}^{\z}_{\Lambda,\beta,\mu_0})''(x)=\beta \mu'(x)=\beta^2[(L^{\z}_{\Lambda,\beta,\mu_0})''(\mu(x))]^{-1}.\nonumber
	\end{equation}
	In this way we have:
	\begin{equation}
		\frac{\partial^m \mathcal{J}^{\z}_{\Lambda,\beta,\mu_0}(x)}{\partial x^m}=\frac{\partial^{m-2}[\beta^2(L^{\z}_{\Lambda,\beta,\mu_0})''(\mu(x))]^{-1}}{\partial x^{m-2}},
		\label{1}
	\end{equation}
	with
	\begin{equation}
		\mu'(x)=\beta [(L^{\z}_{\Lambda,\beta,\mu_0})''(\mu(x))]^{-1}.
		\label{2}
	\end{equation}
	Then, the coefficient $Q_{\Lambda,\z}^{(j)}$ - which is the derivative of order $j$ of $\mathcal{J}_{\Lambda,\beta,\mu_0}^{\z}(x)$ for $x=\bN$ - can be obtained from \eqref{1} and \eqref{2}  taking into account that, when $x=\bN$, the quantities in the right hand side of \eqref{1} and \eqref{2} are given by \eqref{2-mom} and \eqref{m-mom}.
	
	Note that the relations expressed in \eqref{1} and \eqref{2} are the same which exist  between $f_{\beta}(\rho)$ and $p_{\beta}(\mu)$ as well as their grand-canonical finite volume versions ($\fgc(\rho_{\Lambda})$ and $p_{\Lambda,\beta,\z}(\mu)$).
\end{remark}

We conclude the discussion by noting that the formulations expressed in \cite{del1974local} and \cite{dobrushin1994large} are equivalent to our formulation (Theorems \ref{TH-LD}, \ref{TH-MD} and Corollary \ref{CLT}). This is due to the fact that for all $\hat{\rho}_{\Lambda}$ and $\hat{\rho}_{\Lambda}^*$ which satisfy \eqref{MeanV} and \eqref{N*} (with the appropriate chemical potential $\mu(\hat{\rho}_{\Lambda})$), from \eqref{Mu0Rho*} and Remark \ref{RemarkF_GC} we have (\cite{scola2020local}): 
\begin{equation}
	\left|\fgc(\hat{\rho}_{\Lambda})-\F(\hat{\rho}^*_{\Lambda})\right|\le C\frac{\log\sqrt{|\Lambda|}}{|\Lambda|}\nonumber
\end{equation}
and
\begin{equation}
	\left|(\fgc)'(\hat{\rho}_{\Lambda})-\F'(\hat{\rho}^*_{\Lambda})\right|\le C\frac{1}{|\Lambda|}.\nonumber
\end{equation}
Then, defining $I^{C}_{\Lambda,\beta,\z}(\rho_{\Lambda};\rho^*_{\Lambda}):=\beta \F(\rho_{\Lambda})-\beta \F(\rho^*_{\Lambda})+\beta \F'(\rho^*_{\Lambda})(\rho_{\Lambda}-\rho^*_{\Lambda})$ and remembering that $\rho_{\Lambda}=N/|\Lambda|$, from \eqref{FeRho} and \eqref{Mu0Rho*} we have
\begin{eqnarray}
	\left|[I^{C}_{\Lambda,\beta,\z}(\hat{\rho}^*_{\Lambda};\rho^*_{\Lambda})-\beta \F'(\rho^*_{\Lambda})(\hat{\rho}^*_{\Lambda}-\rho^*_{\Lambda})]-[\beta f_{\Lambda,\beta,\z}(\hat{N}^*)+\beta f_{\Lambda,\beta,\z}(N^*)]\right|\nonumber\\
	\le C\frac{\log\sqrt{|\Lambda|}}{|\Lambda|}\nonumber
\end{eqnarray}
as well as
\begin{equation}
	\left|I^{C}_{\Lambda,\beta,\z}(\hat{\rho}^*_{\Lambda}; \rho^*_{\Lambda})-I^{GC}_{\Lambda,\beta,\z}(\rho_{\Lambda};\bRL)\right|\le C_1\frac{\log\sqrt{|\Lambda|}}{|\Lambda|},\nonumber 
\end{equation}
with $C,\;C_1\in\mathbb{R}^+$. Moreover, this equivalence is also true in the thermodynamic limit, which is proved in Sections 3 and 4 of \cite{dobrushin1994large} for the quantities defined in \eqref{H}, \eqref{m-mom} and \eqref{IT}, where in our case it comes from Theorem \ref{TH1} and \cite{scola2020local}. Indeed from Appendix B in   \cite{scola2020local} we have:   
\begin{equation}
	|f^{(m)}_{\beta}(\rho_0)-\F^{(m)}(\rho^*_{\Lambda})|\le C\frac{|\partial\Lambda|}{|\Lambda|},\nonumber
\end{equation}
for all $m\ge0$.

\section*{Acknowledgements}
It is a great pleasure to thank Sabine Jansen, Errico Presutti and Dimitrios Tsagkarogiannis for their generous availability and for assisting the author with many necessary, stimulating and fruitful discussions.

\appendix

\section{Proofs of Theorem \ref{TH-LD}, Theorem \ref{TH-MD} and Corollary \ref{CLT}}
\label{S2}	

We define the following objects:
\begin{equation}
	J^{C}_{\mu}(N,N'):=\frac{e^{\beta\mu N}\canBC(N)}{e^{\beta\mu N'}\canBC(N')},
	\label{NUM-C}
\end{equation}
and
\begin{equation}
	K(\mu,N):=\left(\frac{\GcanBC(\mu)}{e^{\beta\mu N}\canBC(N)}\right)^{-1}.
	\label{Def-Normalizzazione}
\end{equation} 

Let us note that from \eqref{FVfreeE} the term $J^{C}_{\mu_0}(\tN,\bN)$ can be written as: 

\begin{equation}
	J^{C}_{\mu_0}(\tN,\bN)=\exp\left\{ \beta\mu_0(\tN-\bN)+|\Lambda|\beta f_{\Lambda,\beta,\z}(\bN)-|\Lambda|\beta f_{\Lambda,\beta,\z}(\tN)\right\},
	\label{J}
\end{equation}
which is the finite volume version of \eqref{LD-F} viewed in the canonical ensemble. Moreover, we can also write
\begin{equation}\label{KwithJ}
	[K(\mu_0,\bar N_\Lambda)]^{-1}=\sum_{N\ge0}J^{C}_{\mu_0}(N,\bN).
\end{equation} 
Finally, before giving the proofs of the theorems, we remark that the object defined in \eqref{H}, can be written as follows:
\begin{equation}
	L^{\z}_{\Lambda,\beta,\mu_0}(\mu)=\beta|\Lambda|\left[p_{\beta,\Lambda,\z}(\mu+\mu_0)-p_{\beta,\Lambda,\z}(\mu_0)\right].
	\label{relation1}
\end{equation}

\begin{proof}[Proof of Theorem \ref{TH-LD}]		
	We rewrite $\probBC(A_{\tN})$ as follows:
	\begin{equation}\label{option1}
		\probBC(A_{\tN})=\frac{\GcanBC(\tilde{\mu}_{\Lambda})e^{\beta\mu_0\tN}}{\GcanBC(\mu_0)e^{\beta\tilde{\mu}_{\Lambda}\tN}}\mathbb{P}^{\z}_{\Lambda,\tilde{\mu}_{\Lambda}}(A_{\tN}).
	\end{equation}
	In the previous one we did the Radon-Nikod\'ym derivative of our probability measure with respect to the one with $\tilde{\mu}_{\Lambda}$ instead of $\mu_0$. Note that the definition of $\tilde{\mu}_{\Lambda}$ given via \eqref{GcFreeE}, i.e., such that 
	\begin{equation}
		\beta\fgc(\tRL)=\beta \tilde{\mu}_{\Lambda}\tRL-\beta p_{\Lambda,\beta,\z}(\tilde{\mu}_{\Lambda}),
		\label{f-GC_Tmu}
	\end{equation}
	is equivalent to define implicitly $\tilde{\mu}_{\Lambda}$ as the chemical potential such that
	\begin{equation}
		\frac{\tilde{N}}{|\Lambda|}=\mathbb{E}^{\z}_{\Lambda,\tilde{\mu}_{\Lambda}}\left[\frac{N}{|\Lambda|}\right]=\frac{\partial}{\partial\mu}p_{\Lambda,\beta,\z}(\mu)\bigg|_{\mu=\tilde{\mu}_{\Lambda}}.
		\label{Tilde-mu}
	\end{equation}
	Moreover, from \eqref{relation1} and \eqref{FV-I} we have that this $\tilde{\mu}_{\Lambda}$ is equal to the one which satisfies \eqref{I}.

	From \eqref{FVLattice}, \eqref{GcFreeE}, \eqref{LDO} and \eqref{f-GC_Tmu}  we get
	\begin{eqnarray}
		\frac{\GcanBC(\tilde{\mu}_{\Lambda})e^{\beta\mu_0\tN}}{\GcanBC(\mu_0)e^{\beta\tilde{\mu}_{\Lambda}\tN}} &= &\exp\left\{|\Lambda|\left[\beta\mu_0\tilde{\rho}_{\Lambda}-\beta\tilde{\mu}_{\Lambda}\tilde{\rho}_{\Lambda}+\beta p_{\Lambda,\beta,\z}(\tilde{\mu}_{\Lambda})-\beta p_{\Lambda,\beta,\z}(\mu_0)\pm\beta\mu_0\bar{\rho}_{\Lambda}\right]\right\}\nonumber
		\\
		&=&\exp\left\{|\Lambda|\left[\beta \fgc(\bRL)-\beta\fgc(\tRL)+\beta\mu_0(\tRL-\bRL)\right]\right\}\nonumber
		\\
		&=& \exp\left\{- |\Lambda| I^{GC}_{\Lambda,\beta,\z}(\tRL;\bRL)\right\}.
	\end{eqnarray}
	On the other hand, denoting with $\tN^*$ the number of particles such that 
	\begin{equation}
		\sup_{N}\left\{e^{\beta\tilde{\mu}_{\Lambda}N}\canBC(N)\right\}=e^{\beta\tilde{\mu}_{\Lambda}\tN^*}\canBC(\tN^*),
	\end{equation}
	using \eqref{NUM-C} and \eqref{Def-Normalizzazione} we have
	\begin{equation}
		\mathbb{P}^{\z}_{\Lambda,\tilde{\mu}_{\Lambda}}(A_{\tN})=J^{C}_{\tilde{\mu}_{\Lambda}}(\tN,\tN^*)K(\tilde{\mu}_{\Lambda},\tN^*).
	\end{equation}
	The novelty here is that we compute the above term using cluster expansions instead of inverting the characteristic function.
	First, we recall that we have
	\begin{equation}\label{est}
		|\tN-\tN^*|\le C,
	\end{equation}
	for some $C>0$ which does not depend on $\Lambda$.
	Then we find
	\begin{eqnarray}
		J^{C}_{\tilde{\mu}_{\Lambda}}(\tN,\tN^*)&=&\exp\left\{ S'_{|\Lambda|}(\tRL^*)(\tN-\tN^*)-\sum_{m\ge 2}\frac{(\tN-\tN^*)^{m}}{|\Lambda|^{m-1}}\frac{\F^{(m)}(\tRL^*)}{m!}+ |\Lambda|S_{|\Lambda|}(\tRL^*)\right\}\nonumber
		\\
		&\s &\exp\left\{|\Lambda|S_{|\Lambda|}(\tRL^*)\right\}\left(1+\frac{1}{|\Lambda|}\right),
		\label{Eq1LD}
	\end{eqnarray}
	since \eqref{est} and where $S_{|\Lambda|}(\tilde{\rho}^*_{\Lambda})$ is a term of order $\log{\sqrt{|\Lambda|}}/|\Lambda|$ (see Appendix B in \cite{scola2020local}).
	
	The study of $K(\tilde{\mu}_{\Lambda},\tN^*)$ is the same as the one done in Lemma 6.3 of \cite{scola2020local} where now we consider $\tN^*$ as center of fluctuations of order 1/2.
	Hence the conclusion follows from 
	\begin{eqnarray}
		K(\tilde{\mu}_{\Lambda},\tN^*)\le e^{-|\Lambda|S_{|\Lambda|}(\tRL^*)}\left[\sqrt{2\pi D_{\Lambda,\z}(\tRL^*)|\Lambda|}\left(1-\frac{C}{\sqrt{|\Lambda|}}\right)\right]^{-1}
	\end{eqnarray}	
	and
	\begin{equation}
		K(\tilde{\mu}_{\Lambda},\tN^*)\ge e^{-|\Lambda|S_{|\Lambda|}(\tRL^*)}\left[\sqrt{2\pi D_{\Lambda,\z}(\tRL^*)|\Lambda|}\left(1+\frac{C}{\sqrt{|\Lambda|}}\right)\right]^{-1}
	\end{equation}
	for some $C\in\mathbb{R}^+$ independent on $\Lambda$.
\end{proof}

%---------------------------------------------

\begin{proof}[Proof of Theorem \ref{TH-MD}]
	
	From \eqref{NUM-C}, \eqref{Def-Normalizzazione} we have 
	\begin{equation}
		\probBC(A_{\tN})=J^{C}_{\mu_0}(\tN,N^*)K(\mu_0,N^*).
	\end{equation}
	Then using Lemma 6.2 in \cite{scola2020local} we have 
	\begin{equation}
		J^{C}_{\mu_0}(\tN,N^*)\ge\exp\left\{-\frac{(u')^2|\Lambda|^{2\alpha-1}}{2D^{\alpha}_{\Lambda,\beta,\z}(\rho^*_{\Lambda})}+|\Lambda|S_{|\Lambda|}(\rho^*_{\Lambda})-E_{|\Lambda|}(\alpha,u',\rho^*_{\Lambda})\right\}
	\end{equation}
	and
	\begin{equation}
		J^{C}_{\mu_0}(\tN,N^*)\le \exp\left\{-\frac{(u')^2|\Lambda|^{2\alpha-1}}{2D^{\alpha}_{\Lambda,\beta,\z}(\rho^*_{\Lambda})}+|\Lambda|S_{|\Lambda|}(\rho^*_{\Lambda})+E_{|\Lambda|}(\alpha,u',\rho^*_{\Lambda})\right\}
	\end{equation}
	where $S_{|\Lambda|}(\rho^*_{\Lambda})$ is a term of order $\log{\sqrt{|\Lambda|}}/|\Lambda|$ (see Appendix B in \cite{scola2020local}).
	
	The conclusion follows from Lemma 6.3 in \cite{scola2020local} which gives us
	\begin{eqnarray}
		K(\mu_0,N^*)\le e^{-|\Lambda|S_{|\Lambda|}(\rho^*_{\Lambda})}\left[\sqrt{2\pi D^{\alpha,+}_{\Lambda,\z}(\rho^*_{\Lambda})|\Lambda|}\left(1-E_{|\Lambda|}(\alpha,u',\rho^*_{\Lambda})\right)\right]^{-1}
	\end{eqnarray}	
	and
	\begin{equation}
		K(\mu_0,N^*)\ge e^{-|\Lambda|S_{|\Lambda|}(\rho^*_{\Lambda})}\left[\sqrt{2\pi D^{\alpha,+}_{\Lambda,\z}(\rho^*_{\Lambda})|\Lambda|}\left(1+E_{|\Lambda|}(\alpha,u',\rho^*_{\Lambda})\right)\right]^{-1}.
	\end{equation}
\end{proof}

\bibliographystyle{plain}
\bibliography{biblioIsing}
\nocite{*}

\end{document}